\documentclass{tlp}

\usepackage[T1]{fontenc}
\usepackage{amssymb}
\usepackage{amsmath}
\usepackage{hyperref}
\usepackage{enumitem}
\usepackage{fancyvrb}

\newtheorem{algunif}{Unification Algorithm}%{\bfseries}{\normalshape}
\newtheorem{definition}{Definition} % [section]
\newtheorem{example}{Example} % [section]
\newtheorem{lemma}{Lemma} % [section]
\newtheorem{proposition}{Proposition} % [section]
\newtheorem{theorem}{Theorem} % [section]
\newcommand{\ie}{\textit{i.e.}, }
\newcommand{\eg}{\textit{e.g.}, }
\newcommand{\uppereg}{\textit{E.g.}, }
\newcommand*{\wrt}{{w.r.t.}}

\newcommand{\cC}{\mathcal{C}}

\newcommand{\rra}{\mathop{\Rightarrow}}%\limits}
\newcommand{\equi}{\mathop{\thicksim}}

\newcommand{\nat}{\mathbb{N}}
\newcommand{\termset}{T(\Sigma, X)}
\newcommand{\termsett}{T(\Sigma \cup H, X)}
\newcommand{\stermset}{T(\Sigma \cup \Upsilon, X)}

\newcommand{\subsset}{S(\Sigma, X)}
\newcommand{\ssubsset}{S(\Sigma \cup \Upsilon, X)}
\newcommand{\scontext}{\chi^{(1)}}

\newcommand{\binunf}{\operatorname{\mathit{binunf}}}
\newcommand{\patunf}{\operatorname{\mathit{patunf}}}
\newcommand{\binruleset}{\Re}
\newcommand{\patruleset}{\Im}
\newcommand{\calls}{\operatorname{\mathit{calls}}}
\newcommand{\rules}{\operatorname{\mathit{rules}}}

\newcommand{\vars}{\operatorname{\mathit{Var}}}
\newcommand{\dom}{\operatorname{\mathit{Dom}}}
\newcommand{\ran}{\operatorname{\mathit{Ran}}}
\newcommand{\mgu}{\operatorname{\mathit{mgu}}}
\newcommand{\idsub}{\emptyset}%{\varnothing}
\newcommand{\id}{\operatorname{\mathit{id}}}
\newcommand{\patterm}[3]{{{#1} \star {#2} \star {#3}}}
\newcommand{\pattermm}[2]{{{#1} \star {#2}}}
\newcommand{\patsub}[2]{{{#1} \star {#2}}}

\newcommand{\patid}{\operatorname{\mathit{patid}}}
\newcommand{\pt}[1]{{#1}^{\star}}
\newcommand{\qe}{\pt{\emptyseq}}

\newcommand{\fsym}{\operatorname{\mathsf{f}}}

\newcommand{\ssym}{\operatorname{\mathsf{s}}}
\newcommand{\while}{\operatorname{\mathsf{while}}}
\newcommand{\mle}{\operatorname{\mathsf{le}}}
\newcommand{\gt}{\operatorname{\mathsf{gt}}}
\newcommand{\add}{\operatorname{\mathsf{add}}}
\newcommand{\cons}{\operatorname{\mathsf{cons}}}
\newcommand{\nil}{\operatorname{\mathsf{nil}}}
\newcommand{\islist}{\operatorname{\mathsf{isList}}}
\newcommand{\zero}{\mathsf{0}}
\newcommand{\one}{\mathsf{1}}

\newcommand{\sequence}[1]{\left\langle{#1}\right\rangle}
\newcommand{\bigsequence}[1]{\big\langle{#1}\big\rangle}
\newcommand{\seqset}[1]{{\overline{{#1}}}}
\newcommand{\emptyseq}{{\mathsf{e}}}
\newcommand{\nti}{\textsf{NTI}}

\VerbatimFootnotes

\begin{document}

\lefttitle{Cambridge Author}

\jnlPage{1}{25}
\jnlDoiYr{2025}
\doival{10.1017/xxxxx}

\title[Non-Termination of Logic Programs Using Patterns]
{Non-Termination of Logic Programs Using Patterns%
\thanks{This work was partially funded
by the European Program FEDER/INTERREG VI 2021-2027
(grant 2024-1443-004977).}}

\begin{authgrp}
    \author{\sn{Etienne} \gn{Payet}}
    \affiliation{LIM - Université de la Réunion, France \\\email{etienne.payet@univ-reunion.fr}}
\end{authgrp}
%\author{Etienne Payet\inst{1}\orcidID{0000-0002-3519-025X}}
%
%\authorrunning{E. Payet}
%
% \institute{LIM - Université de la Réunion, France\\
% \email{etienne.payet@univ-reunion.fr}\\
% \url{http://lim.univ-reunion.fr/staff/epayet/}}

\maketitle

\begin{abstract}
    In this paper, we consider an approach introduced
    in term rewriting for the automatic detection of
    non-looping non-termination from \emph{patterns}
    of rules. We adapt it to logic programming by
    defining a new unfolding technique that
    produces patterns describing possibly infinite
    sets of finite rewrite sequences. We present an
    experimental evaluation of our contributions that
    we implemented in our tool \nti.
\end{abstract}


%%%%%%%%%%%%%%%%%%%%%%%
\section{Introduction}
%%%%%%%%%%%%%%%%%%%%%%%
This paper is concerned with non-termination in
logic programming, where one rewrites finite 
sequences of terms (called \emph{queries})
according to the operational semantics described,
\eg by~\cite{apt97}.
Rewriting is formalised by binary relations
$\rra_r$ indexed by rules $r$ from the logic
program under consideration and non-termination
by the existence of an infinite rewrite sequence
$Q_0 \rra_{r_1} Q_1 \rra_{r_2} \cdots$
(where the $Q_i$s are queries).
Our motivations are theoretical (study remarkable
forms of infinite rewrite sequences) and practical
(help programmers to detect bugs by providing queries
that run forever).

Most papers related to this topic provide
necessary or sufficient conditions
for the existence of \emph{loops}, \ie finite
rewrite sequences
$Q_0 \rra_{r_1} \cdots \rra_{r_n} Q_n$
where $Q_n$ satisfies a condition $\cC$
that entails the possibility of starting
again, \ie
$Q_n \rra_{r_1} \cdots \rra_{r_n} Q_{2n}$
holds and $Q_{2n}$ also satisfies $\cC$,
and so on.
For example, \cite{payetM06} present a
sufficient condition, based on \emph{neutral}
argument positions of predicate symbols,
which is applied to elements of the
\emph{binary unfolding}
(a set of rules that exhibits the
termination properties of logic programs,
see the paper by~\cite{codishT99}).

In this paper, we are rather interested in
\emph{non-looping} non-termination, \ie
infinite rewrite sequences that do not embed
any loop.
The non-periodic nature of such sequences
makes them difficult to detect, while they
can be produced from simple logic programs,
as those used in our experiments
(Sect.~\ref{sect:experiments}).
We are inspired by the approach
of~\cite{emmesEG12}, introduced in the 
context of term rewriting\footnote{The
operational semantics of term rewriting differs
from that of logic programming: a crucial
difference is that the rewrite relation of
term rewriting is based on instantiation
while that of logic programming relies
on unification.}.
This approach considers \emph{pattern terms}
(``abstract'' terms describing possibly
infinite sets of concrete terms) as well as
\emph{pattern rules} built from pattern
terms. From the term rewrite system under
analysis, it produces pattern rules that are
\emph{correct}, \ie they describe sets of
finite rewrite sequences \wrt{} the
operational semantics of term rewriting.
Nine inference rules are provided to derive
correct pattern rules, as well as a strategy
for their automated application and a sufficient
condition to detect non-looping non-termination. 

We adapt this approach to logic programming.
Our main contributions are:
(i) the definition of a new unfolding
technique that produces correct pattern rules
\wrt{} the operational semantics of logic
programming (this gives a more compact
presentation than  the nine inference rules
of~\cite{emmesEG12} and we do not need
application strategies);
(ii) the definition of a restricted form of
pattern terms, called \emph{simple}, for
which we provide a unification algorithm
(needed to compute the unfolding)
that we prove correct;
(iii) an easily automatable sufficient
condition to detect non-looping non-termination
from pattern rules built from simple pattern
terms;
(iv) the implementation of a non-termination
approach based on these notions 
in our tool \nti, that we have evaluated on
logic programs resulting from the translation
of term rewrite systems used in the
experiments of~\cite{emmesEG12}.
As far as we know, our approach is the
first capable of proving non-termination
of these logic programs automatically.

The paper is organised as follows.
Sect.~\ref{sect:prel} introduces basic
definitions and notations
(a running example illustrating our
contributions starts from Sect.~\ref{sect:binunf}),
Sect.~\ref{sect:patterns} presents our
adaptation of patterns to logic programming,
Sect.~\ref{sect:simple-patterns}
considers the notion of simple pattern,
Sect.~\ref{sect:experiments} presents an
experimental evaluation,
Sect.~\ref{sect:rel-work} describes
related work and
Sect.~\ref{sect:conclusion} concludes with
future work.

%%%%%%%%%%%%%%%%%%%%%%%
\section{Preliminaries}
\label{sect:prel}
%%%%%%%%%%%%%%%%%%%%%%%

We let $\nat$ denote the set of natural numbers.
Let $A$ be a set. Then, $\seqset{A}$ is the set
of finite sequences of elements of $A$, which
includes the empty sequence, denoted
as $\emptyseq$.
We use the delimiters $\langle$ and
$\rangle$ for writing elements of
$\seqset{A}$ and juxtaposition to denote
the concatenation operation,
\eg $\sequence{a_0,a_1}\sequence{a_2,a_3} =
\sequence{a_0,a_1,a_2,a_3}$.
We generally denote elements
of $\seqset{A}$ using lowercase letters
with an overline, \eg $\seqset{a}$.

\subsection{Binary Relations}\label{sect:binary-relation}
%%%%%%%%%%%%%%%%%%%%%%%%%%%%%
% A \emph{binary relation} $\psi$ from a set $A$
% to a set $B$ is a subset of $A \times B$.
% For all $a \in A$, we let
% $\psi(a) = \{b \in B \mid (a,b) \in \psi\}$.
% For binary relations that have the form of an
% arrow, we usually write $a \mathop{\psi} b$
% instead of $(a,b) \in \psi$.
% %
% A \emph{function} $f$ from $A$ to $B$ is a
% binary relation from $A$ to $B$ which is
% such that, for all $a \in A$, $f(a)$ is a
% singleton $\{b\}$; we then write
% $f(a) = b$ instead of $f(a) = \{b\}$.

A binary relation $\phi$ on a set $A$ is a
subset of $A^2 = A \times A$. For all
$\varphi \subseteq A^2$, the \emph{composition}
of $\phi$ and $\varphi$ is
$\phi \circ \varphi = \left\{(a,a') \in A^2
\;\middle\vert\;
\exists a_1 \in A: \ (a,a_1) \in \phi \land
(a_1,a') \in \varphi \right\}$.
We let $\phi^0$ be the identity relation and,
for any $n\in\nat$, $\phi^{n+1}=\phi^n \circ \phi$.
Moreover, $\phi^+ = \bigcup \left\{\phi^n \mid n > 0 \right\}$
(resp. $\phi^* = \phi^0 \mathop{\cup} \phi^+$)
is the transitive (resp. reflexive and
transitive) \emph{closure} of $\phi$.
A \emph{$\phi$-chain} (or \emph{chain} if
$\phi$ is clear from the context) is a (possibly
infinite) sequence of elements of $A$ such that
$(a,a') \in \phi$ for any two consecutive elements
$a,a'$.
% (hence, the empty sequence and the
% singletons are chains).
For binary relations that have the form of an
arrow, \eg $\rra$, we may write chains
$a_0,a_1,\dots$ as $a_0 \rra a_1 \rra \cdots$.

\subsection{Terms and Substitutions}
\label{sect:term-subs}
%%%%%%%%%%%%%%%%%%
We use the same definitions and notations as~\cite{baaderN98}
for terms.
A \emph{signature} is a set of \emph{function symbols}, each
element of which has an \emph{arity} in $\nat$
%, which is the number of its arguments.
(the $0$-ary elements are called \emph{constant symbols}).
We denote function symbols by words in the
\emph{sans serif} font,
\eg $\fsym$, $\zero$, $\while$\dots

Let $\Sigma$ be a signature and $X$ be a set of
\emph{variables} disjoint from $\Sigma$.
For all $m \in \nat$, we let $\Sigma^{(m)}$
denote the set of all $m$-ary elements of $\Sigma$.
The set $\termset$ of all \emph{$\Sigma$-terms over $X$}
(or simply \emph{terms} if $\Sigma,X$ are clear
from the context) is defined as: $X \subseteq \termset$
and, for all $m\in\nat$, all $\fsym \in \Sigma^{(m)}$
and all $s_1,\dots,s_m \in \termset$,
$\fsym(s_1,\dots,s_m) \in \termset$.
For all $s \in \termset$, we let $\vars(s)$ denote
the set of variables occurring in $s$.
% ; this is extended to finite sequences, \ie
% $\vars(\sequence{s_1,\dots,s_m}) =
% \vars(s_1,\dots,s_m) =
% \bigcup \{\vars(s_i) \mid 1 \leq i \leq m\}$.
%
We use the superscript notation to
denote several successive applications of a
unary function symbol, \eg $\ssym^3(\zero)$
is a shortcut for $\ssym(\ssym(\ssym(\zero)))$
and $\ssym^0(\zero) = \zero$.

A \emph{$\termset$-substitution} (or simply
\emph{substitution} if $\termset$ is clear
from the context) is a function $\theta$
from $X$ to $\termset$ such that $\theta(x) \neq x$
for only finitely many variables $x$.
The \emph{domain} of $\theta$ is 
$\dom(\theta)=\{x \in X \mid \theta(x) \neq x\}$.
We let $\ran(\theta) =
\bigcup \{\vars(\theta(x)) \mid x \in \dom(\theta)\}$
and $\vars(\theta) = \dom(\theta) \cup \ran(\theta)$.
We usually write $\theta$ as
$\{x_1\mapsto\theta(x_1), \dots, x_m\mapsto\theta(x_m)\}$
where $\{x_1,\dots,x_m\} = \dom(\theta)$
(hence, the identity substitution is written
as $\idsub$).
%
% Given a set of variables $V$, we denote by
% $\theta | V$ the substitution obtained from
% $\theta$ by restricting its domain to $V$.
%
% We let $\dom(\theta_1,\dots,\theta_m) =
% \bigcup \{\dom(\theta_i) \mid 1 \leq i \leq m\}$
% for all substitutions $\theta_1,\dots,\theta_m$.
%
A \emph{(variable) renaming} is a substitution that
is a bijection on $X$.
We let $\subsset$ denote the set of all
$\termset$-substitutions.

The application of $\theta \in \subsset$ to
$s \in \termset$, denoted as $s\theta$, is
defined as: $s\theta = \theta(s)$ if $s\in X$ and
$s\theta = \fsym(s_1\theta,\dots,s_m\theta)$
if $s = \fsym(s_1,\dots,s_m)$.
Then, $s\theta$ is called an \emph{instance}
of $s$. Application % of $\theta$
is extended to finite sequences of terms:
$\sequence{s_1,\dots,s_m}\theta = 
\sequence{s_1\theta,\dots,s_m\theta}$.

The \emph{composition} of $\sigma,\theta \in \subsset$
is the $\termset$-substitution denoted as
$\sigma\theta$ and defined as: for all $x \in X$,
$\sigma\theta(x) = (\sigma(x))\theta$.
This is an associative operation, \ie
for all $s \in \termset$,
$(s\sigma)\theta = s(\sigma\theta)$.
We say that $\sigma$ \emph{commutes} with
$\theta$ if $x\sigma\theta = x\theta\sigma$
for all $x \in X$.
We say that $\sigma$ is \emph{more general}
than $\theta$ if $\theta=\sigma\eta$ for some
$\eta \in \subsset$.

Let $s,t \in \termset$.
We say that $s$ \emph{unifies} with $t$
(or $s$ and $t$ unify) if $s\sigma=t\sigma$
for some $\sigma \in \subsset$. Then, $\sigma$
is a \emph{unifier} of $s$ and $t$.
% Moreover, $\sigma$ is called a
% \emph{most general unifier} (mgu) of $s$ and $t$
% if it is a unifier of $s$ and $t$ that is more
% general than all unifiers of $s$ and $t$.
We let $\mgu(s,t)$ denote the set of
\emph{most general unifiers} of $s$ and $t$.
All this is naturally extended
to finite sequences of terms.

\subsection{The Signature Used in the Paper}
\label{sect:sig}
%%%%%%%%%%%%%%%%%%%%%%%%%%
We regard the symbol $\emptyseq$ denoting the
empty sequence as a special constant symbol.
To simplify the statements of this paper,
from now on we fix a signature $\Sigma$
and a set
$H = \{\square_n \mid n \in
\nat \setminus \{0\}\}$
of constant symbols (called \emph{holes})
such that $\Sigma$, $\{\emptyseq\}$ and $H$
are disjoint from each other.
We also fix an infinite countable set
$X$ of variables disjoint from
$\Sigma \cup \{\emptyseq\} \cup H$.
A \emph{term} is an element of $\termset$ and
most of the time a substitution is an element
of $\subsset$.
Let $n$ be a positive integer. An
\emph{$n$-context} is an element of
$\termsett$ that contains occurrences of
$\square_1$, \dots, $\square_n$ but no
occurrence of another hole. For all 
$n$-contexts $c$ and all
$s_1,\dots,s_n \in \termsett$,
we let $c(s_1,\dots,s_n)$ denote the
element of $\termsett$ obtained from $c$
by replacing all the occurrences of
$\square_i$ by $s_i$, for all $1 \leq i \leq n$.
We use the superscript notation for denoting
several successive embeddings of a 1-context
$c$ into itself: $c^0 = \square_1$ and, for
all $n \in \nat$, $c^{n+1} = c(c^n)$.
We denote by $\scontext$ the set of
1-contexts that contain no variable.
Terms are generally denoted by $a,s,t,u,v$,
variables by $x,y,z$ and contexts by $c$,
possibly with subscripts and primes.

\subsection{Logic Programming}
\label{sect:lp}
%%%%%%%%%%%%%%%%%%%%%%%%%%
We refer to~\cite{apt97} for the basics of
logic programming. To simplify our presentation,
we place ourselves in the general framework of
\emph{term reduction systems}, \ie we do not
distinguish \emph{predicate/function} symbols,
\emph{terms/atoms}\dots{} and we do not always
use the standard terminology and notations of
logic programming (\eg \emph{rule} instead of
\emph{clause}).
\begin{definition}
    A \emph{program} is a subset of
    $\termset \times \seqset{\termset}$,
    every element of which is called a \emph{rule}.
    A rule $(u,\seqset{v})$ is \emph{binary}
    if $\seqset{v}$ is empty or is a singleton.
    We let $\binruleset$ denote the set of
    binary rules. For the sake of readability,
    we omit the delimiters $\langle$ and $\rangle$
    in the right-hand side of a binary rule,
    which amounts to considering that 
    $\binruleset \subseteq
    \termset \times (\termset \cup \{\emptyseq\})$.
\end{definition}

Given a rule $(u,\seqset{v})$, we let
$\left[(u,\seqset{v})\right] =
\left\{(u\gamma,\seqset{v}\gamma) \mid \gamma 
\text{ is a renaming}\right\}$ denote its
\emph{equivalence class modulo renaming}.
For all sets of rules $U$, we let 
$[U] = \bigcup_{r \in U} [r]$.
Moreover, for all rules or sequences of
terms $S$, we write
$\seqset{r} \ll_S U$ to denote that
$\seqset{r}$ is a sequence of elements
of $U$ variable disjoint from $S$ and 
from each other.

The rules of a program allow one to rewrite
finite sequences of terms. This is formalised
by the following binary relation, which
corresponds to the operational semantics of
logic programming with the leftmost selection
rule.
\begin{definition}\label{def:rewrite-rel}
    For all programs $P$, we let
    $\rra_P =
    \bigcup \left\{\rra_r \;\middle\vert\;
    r\in P \right\}$
    where, for all $r \in P$,
    \begin{align*}
        \rra_r =& \left\{
            \big(\sequence{s}\seqset{s},
            (\seqset{v}\,\seqset{s})\theta\big)
            \in \seqset{\termset}^2
            \;\middle\vert\;
            \begin{array}{l}
                \sequence{(u, \seqset{v})}
                \ll_{\sequence{s}\seqset{s}} [r], \
                \theta \in \mgu(u, s)
            \end{array}\right\}
    \end{align*}
    For all $s \in \termset$,
    $\calls_P(s) =
    \{t \in \termset \mid 
    \sequence{s} \rra_P^+ \sequence{t,\dots}\}
    \cup \{\emptyseq \mid
    \sequence{s} \rra_P^+ \emptyseq\}$
    is the set of
    \emph{calls in the $\rra_P$-chains that
    start from $s$}.
\end{definition}

\subsection{Binary Unfolding}
\label{sect:binunf}
%%%%%%%%%%%%%%%%%%%%%%%%%%
The binary unfolding of a program $P$
(see the paper by ~\cite{codishT99})
is a set of binary rules, denoted as
$\binunf(P)$, that captures call patterns
of $P$.
It corresponds to the transitive closure
of a binary relation which relates
consecutive calls selected in a computation
(Prop.~\ref{prop:binunf_calls} below).
Non-termination for a specific sequence of
terms implies the existence of a
corresponding infinite chain in this relation
(Thm.~\ref{thm:observing-termination-lp}
below).

More precisely, $\binunf(P)$
is defined as the least fixed point of a
function $T_P^{\beta}$ on the power set of $\binruleset$.
For all $U \subseteq \binruleset$, $T_P^{\beta}(U)$
is constructed by unfolding prefixes of right-hand
sides of rules from $P$ using $U$. Let
$(u,\sequence{v_1,\dots,v_m}) \in P$:
\begin{enumerate}[label=(\roman*)]
    \item \label{T_P_beta_i}
    for each $1 \leq i \leq m$, one
    unfolds $v_1,\dots,v_{i-1}$ with
    $(u_1,\emptyseq),\dots,(u_{i-1},\emptyseq)$
    from $U$ to obtain a corresponding
    instance of $(u,v_i)$,
    \item for each $1 \leq i \leq m$, one
    unfolds $v_1,\dots,v_{i-1}$ with
    $(u_1,\emptyseq),\dots,(u_{i-1},\emptyseq)$
    from $U$ and $v_i$ with
    $(u_i,v)$ from $U$ to obtain
    a corresponding instance of $(u,v)$,
    \item one unfolds $v_1,\dots,v_m$
    with $(u_1,\emptyseq),\dots,(u_m,\emptyseq)$
    from $U$ to obtain a corresponding
    instance of $(u,\emptyseq)$.
\end{enumerate}

This is formally expressed as follows, using the
set $\id$ of \emph{identity binary rules},
which consists of every pair
$(\fsym(x_1,\dots,x_m),\fsym(x_1,\dots,x_m))$
where $\fsym \in \Sigma^{(m)}$
and $x_1,\dots,x_m$ are distinct variables.
The use of $\id$ allows one to cover
case~\ref{T_P_beta_i} above.

\newpage
\begin{definition}\label{def:binunf}
    For all programs $P$ and all
    $U \subseteq \binruleset$, we let
    \[T_P^{\beta}(U) =
    \big[(u,\emptyseq) \in P\big] \cup 
    \left[\left(u\theta, v\theta\right)
    \; \middle| \;
    \begin{array}{l}
        r = (u,\sequence{v_1,\dots,v_m}) \in P,\
        1 \leq i \leq m \\
        \sequence{(u_1,\emptyseq),
        \dots,
        (u_{i-1},\emptyseq), (u_i,v)}
        \ll_r U \cup \id \\
        \text{if $i < m$ then $v \neq \emptyseq$} \\
        \theta \in \mgu\big(\sequence{u_1,\dots,u_i},
        \sequence{v_1,\dots,v_i}\big)
    \end{array}
    \right]\]
    The \emph{binary unfolding of $P$} is the
    set of binary rules
    $\binunf(P) = \big(T_P^{\beta}\big)^*(\emptyset)$.
\end{definition}

Intuitively, each $(u,v) \in \binunf(P)$
specifies that some instance of $v$ belongs
to $\calls_P(u)$. More generally, we have:
\begin{proposition}\label{prop:binunf_calls}
    Let $P$ be a program,
    $(u,v) \in \binunf(P)$ and
    $\sigma \in \subsset$.
    Then, for some $\theta \in \subsset$,
    we have $v\theta \in \calls_P(u\sigma)$.
\end{proposition}

\begin{example}\label{ex:binunf}
    Let $P$ be the program which consists of the rules
    \[{\renewcommand{\arraystretch}{1.1}
    \begin{array}{rcl@{\qquad}rcl}
        r_1 & = & \multicolumn{4}{l}{\left(\while(x, y),
        \sequence{\gt(x, y), \add(x, y, z),
        \while(z, \ssym(y))}\right)} \\
        r_2 & = & \left(\gt(\ssym(x), \zero), \emptyseq\right) &
        r_3 & = & \left(\gt(\ssym(x), \ssym(y)), \gt(x, y)\right) \\
        r_4 & = & \left(\add(x, \zero, x),\emptyseq\right) &
        r_5 & = & \left(\add(x, \ssym(y), \ssym(z)),
        \add(x, y, z)\right) \\
        r_6 & = & \left(\while(x, y), \mle(x, y)\right) & & & \\
        r_7 & = & \left(\mle(\zero, x), \emptyseq\right) &
        r_8 & = & \left(\mle(\ssym(x), \ssym(y)), \mle(x, y)\right) \\
    \end{array}}\]
    and which corresponds to the imperative program
    fragment
    \begin{center}
        \verb-while (x > y) { x = x + y; y = y + 1; }-
    \end{center}
    Rule $r_1$ is used to continue the loop and
    $r_6$ is used to stop it.
    Note that this imperative fragment does not
    terminate if it is run from integers $x,y$ such that
    $x > y > 0$.
    
    Let us compute some elements of $\binunf(P)$
    by applying Def.~\ref{def:binunf}.
    \begin{itemize}
        \item Obviously, we have
        $[r_2] \cup [r_4] \subseteq
        T_P^{\beta}(\emptyset)$.
        \item Let us unfold the whole right-hand
        side of $r_1$, \ie let us consider
        $i = m = 3$. We have
        \begin{align*}
            \big\langle
            \left(\gt(\ssym(x_1), \zero),
            \emptyseq\right),
            & \left(\add(x_2, \zero, x_2),
            \emptyseq\right), \\
            & \left(\while(x_3, y_3),
            \while(x_3, y_3)\right) \big\rangle
            \ll_{r_1} [r_2] \cup [r_4] \cup \id
        \end{align*}
        and
        $\theta = \{x \mapsto \ssym(x_1),
        y \mapsto \zero,
        z \mapsto \ssym(x_1),
        x_2 \mapsto \ssym(x_1),
        x_3 \mapsto \ssym(x_1),
        y_3 \mapsto \ssym(\zero)\}$
        is the mgu of %the sequences
        $\sequence{\gt(\ssym(x_1), \zero),
        \add(x_2, \zero, x_2),
        \while(x_3, y_3)}$ and
        $\sequence{\gt(x, y), \add(x, y, z),
        \while(z, \ssym(y))}$.
        Consequently, the set $(T_P^{\beta})^2(\emptyset)$
        contains the binary rule
        $\left(\while(x, y)\theta,
        \while(x_3, y_3)\theta\right) =
        \left(\while\left(\ssym(x_1), \zero\right),
        \while\left(\ssym(x_1), \ssym(\zero)\right)\right)$.
        \item More generally, we have
        $[r'_n \mid n \in \nat] \subseteq \binunf(P)$
        where, for all $n \in \nat$,
        $r'_n =
        \left(\while(\ssym^{n+1}(x),\ssym^n(\zero)),
        \while(\ssym^{2n+1}(x),\ssym^{n+1}(\zero))
        \right)$.
        % \qed
    \end{itemize}
\end{example}

The binary unfolding exhibits the termination
properties of a program:
\begin{theorem}[see the paper by~\cite{codishT99}]
    \label{thm:observing-termination-lp}
    Let $P$ be a program and $\seqset{s}$
    be a sequence of terms.
    Then, there is an infinite $\rra_P$-chain
    that starts from $\seqset{s}$ 
    \emph{if and only if} 
    there is an infinite
    $\rra_{\binunf(P)}$-chain
    that starts from $\seqset{s}$.
\end{theorem}

\begin{example}[Ex.~\ref{ex:binunf} cont.]
    \label{ex:binunf2}
    For all $n > m > 0$, we have the infinite
    $\rra_P$-chain
    \begin{align*}
        \sequence{\while(\ssym^n(\zero), \ssym^m(\zero))}
        & \mathop{\big(\rra_{r_1} \circ \rra_{r_3}^m \circ \rra_{r_2}
        \circ \rra_{r_5}^m \circ \rra_{r_4}\big)}
        \sequence{\while(\ssym^{n+m}(0), \ssym^{m+1}(0))} \\
        & \mathop{\big(\rra_{r_1} \circ \rra_{r_3}^{m+1} \circ \rra_{r_2}
        \circ \rra_{r_5}^{m+1} \circ \rra_{r_4}\big)}
        \cdots
    \end{align*}
    and also the infinite $\rra_{\binunf(P)}$-chain
    \[\sequence{\while(\ssym^n(\zero), \ssym^m(\zero))}
    \mathop{\rra_{r'_m}}
    \sequence{\while(\ssym^{n+m}(0), \ssym^{m+1}(0))}
    \mathop{\rra_{r'_{m+1}}} \cdots\]
    We note that none of these chains embeds a loop:
    in the $\rra_P$-chain, the number of applications
    of $r_3$ and $r_5$ gradually increases and,
    in the $\rra_{\binunf(P)}$-chain, a new binary
    rule (not occurring before) is used at each step.
    % \qed
\end{example}

%%%%%%%%%%%%%%%%%%%%%%%%%%%%%%%%%
\section{Patterns}
\label{sect:patterns}
%%%%%%%%%%%%%%%%%%%%%%%%%%%%%%%%%

In this section, we describe our adaptation 
to logic programming of the pattern approach
introduced by~\cite{emmesEG12}.
Our main idea is similar to that
of~\cite{payetM06}, \ie unfold the program
and try to prove its non-termination from
the resulting set. To this end, based on the
binary unfolding mentioned previously
(Def.~\ref{def:binunf}), we introduce
a new unfolding technique that produces patterns
of rules (Def.~\ref{def:patunf}) and a
sufficient condition to non-termination that
we apply to the generated patterns
(Thm.~\ref{theo:detection-nonterm}).

First, we recall the definition of pattern
term and pattern rule, that we formulate
differently from~\cite{emmesEG12} to fit
our needs. In particular, we introduce the
concept of pattern substitution.
\begin{definition}\label{def:pattern-subs}
    A \emph{pattern substitution} is a pair
    $\theta = (\sigma,\mu) \in \subsset^2$,
    rather denoted as $\patsub{\sigma}{\mu}$.
    % and we call $\sigma$ the \emph{pumping substitution}
    % and $\mu$ the \emph{closing substitution}.
    %
    For all $n \in \nat$, we let $\theta(n) = \sigma^n\mu$.
    We say that $\theta$ \emph{describes} the set 
    $\left\{\theta(n) \;\middle\vert\;
    n \in \nat\right\} \subseteq \subsset$.
\end{definition}
For instance, if $\sigma = \{x \mapsto \ssym(x),
y \mapsto \ssym(y)\}$ and
$\mu = \{x \mapsto \ssym(x), y\mapsto \zero\}$
then $\theta = \patsub{\sigma}{\mu}$ is a pattern
substitution. For all $n \in \nat$, we have
$\theta(n) = \sigma^n\mu =
\{x\mapsto \ssym^{n+1}(x), y\mapsto \ssym^n(\zero)\}$.

From pattern substitutions, we define pattern terms.
\begin{definition}\label{def:pattern-term}
    A \emph{pattern term} is a pair $p = (s,\theta)$
    where $s \in \termset$ and $\theta$ is a pattern
    substitution. We denote it as $\pattermm{s}{\theta}$
    or $\patterm{s}{\sigma}{\mu}$
    if $\theta = \patsub{\sigma}{\mu}$.
    For all $n \in \nat$, we let
    $p(n) = s\theta(n)$.
    We say that $p$ \emph{describes} the set
    $\left\{p(n) \;\middle\vert\;
    n \in \nat\right\} \subseteq \termset$.
    For all $s \in \termset$, we let
    $\pt{s} = \patterm{s}{\idsub}{\idsub}$.
\end{definition}
For instance, $p = \patterm{\gt(x,y)}
{\{x \mapsto \ssym(x), y \mapsto \ssym(y)\}}
{\{x \mapsto \ssym(x), y\mapsto \zero\}}$
is a pattern term. For all $n \in \nat$, we
have
$p(n) = \gt(\ssym^{n+1}(x),\ssym^n(\zero))$.

Then, from pattern terms one can define pattern
rules.
\begin{definition}\label{def:pattern-rule}
    A \emph{pattern rule} is a pair $r = (p,q)$
    of pattern terms.
    It describes the set of binary rules
    $\rules(r) = \left\{(p(n),q(n))
    \;\middle\vert\; n \in \nat \right\}$.
    We let $\patruleset$ denote the set
    of pattern rules.
\end{definition}

\newpage
\begin{example}[Ex.~\ref{ex:binunf} cont.]
    \label{ex:running}
    Let $u = \while(x,y)$ be the left-hand
    side of $r_1$,
    $\sigma = \{x\mapsto \ssym(x),
    y\mapsto \ssym(y)\}$,
    $\sigma' = \{x\mapsto \ssym(x)\}$
    and
    $\mu = \{x\mapsto \ssym(x),
    y\mapsto \zero\}$.
    The pattern terms
    \begin{align*}
        p &= \patterm{u\sigma}{\sigma}{\mu} =
        \patterm{\while(\ssym(x),\ssym(y))}{\sigma}{\mu} \\
        q &= \patterm{u\sigma^2}{\sigma\sigma'}{\mu}
        = \patterm{\while(\ssym^2(x),\ssym^2(y))}
        {\{x\mapsto \ssym^2(x), y\mapsto \ssym(y)\}}
        {\mu}
    \end{align*}
    respectively describe the sets of terms
    $\left\{p(n) =
    \while\left(\ssym^{n+2}(x),\ssym^{n+1}(\zero)\right)
    \;\middle\vert\;
    n \in \nat\right\}$ and
    $\left\{q(n) =
    \while\left(\ssym^{2n+3}(x),\ssym^{n+2}(\zero)\right)
    \;\middle\vert\;
    n \in \nat\right\}$.
    Moreover,
    \begin{align*}
        \rules((p,q)) &= 
        \left\{\left(\while(\ssym^{n+2}(x),\ssym^{n+1}(\zero)),
        \while(\ssym^{2n+3}(x),\ssym^{n+2}(\zero))\right)
        \;\middle\vert\; n \in \nat\right\} \\
        &= \left\{r'_n \mid n > 0\right\}
        \subseteq \left\{r'_n \mid n \in \nat\right\}
        \subseteq \binunf(P)
        \text{ (see Ex.~\ref{ex:binunf})}
    \end{align*}
\end{example}

The notion of \emph{correctness} of a pattern
rule is defined by~\cite{emmesEG12} in
the context of term rewriting. We reformulate
it as follows in logic programming.

\begin{definition}\label{def:correct}
    Let $P$ be a program.
    A pattern rule $r$ is \emph{correct}
    \wrt{} $P$ if
    $\rules(r) \subseteq \binunf(P)$.
    A set $U$ of pattern rules is correct
    \wrt{} $P$ if all its elements are.
\end{definition}

So, if a pattern rule $(p,q)$ is correct
\wrt{} $P$ then, for all $n \in \nat$, we
have $(p(n),q(n)) \in \binunf(P)$, \ie by
Prop.~\ref{prop:binunf_calls}, 
$\sequence{p(n)} \rra^+_P \sequence{q(n)\theta,\dots}$
for some $\theta \in \subsset$.
Intuitively, this means that for all $n \in \nat$,
a call to $p(n)$ necessarily leads to a call to
$q(n)$.
For instance, in Ex.~\ref{ex:running}, we have
$\rules((p,q)) \subseteq \binunf(P)$, hence
$(p,q)$ is correct \wrt{} $P$ and we have
$\sequence{\while(\ssym^{n+2}(x),\ssym^{n+1}(\zero))}
\rra_P^+
\sequence{\while(\ssym^{2n+3}(x),\ssym^{n+2}(\zero))}$
for all $n \in \nat$.

The next result allows one to infer correct 
pattern rules from a program. It considers
pairs of rules that have the same form as
$(r_2,r_3)$, $(r_4,r_5)$ and $(r_7,r_8)$
in Ex.~\ref{ex:binunf}.
It uses the set of contexts $\scontext$
(see Sect.~\ref{sect:sig}).

\begin{proposition}
    \label{prop:correct-set}
    Suppose that a program $P$ contains
    two binary rules $r = (u,v)$ and 
    $r' = (u', \emptyseq)$ such that
    \begin{itemize}
        \item $u = c(c_1(x_1),\dots,c_m(x_m))$,
        $v = c(x_1,\dots,x_m)$ and
        $u' = c(t_1,\dots,t_m)$,
        \item $\{c_1,\dots,c_m\} \subseteq \scontext$
        and $c$ is an $m$-context with
        $\vars(c) = \emptyset$,
        \item $x_1,\dots,x_m$ are distinct
        variables and $t_1,\dots,t_m$ are terms.
    \end{itemize}
    Then, $(p,\qe)$ and $(q,\pt{v})$ are correct
    \wrt{} $P$ where
    $p = \patterm{v}{\sigma}{\mu}$,
    $q = \patterm{u}{\sigma}{\idsub}$ and
    \begin{align*}
        \sigma &= \{x_k \mapsto c_k(x_k)
        \mid 1 \leq k \leq m,\
        c_k(x_k) \neq x_k \} \\
        \mu &=
        \{x_k \mapsto t_k \mid 1 \leq k \leq m,\
        t_k \neq x_k\}
    \end{align*}
\end{proposition}

\begin{example}\label{ex:correct}
    Let us regard $(r_2,r_3)$ and
    $(r_4,r_5)$ from Ex.~\ref{ex:binunf}.
    \begin{itemize}
        \item % We have
        $r_2 = (c(t_1,t_2),\emptyseq)$ and
        $r_3 = \big(c(c_1(x),c_1(y)), c(x,y)\big)$
        for 
        $t_1 = \ssym(x)$, $t_2 = \zero$,
        $c = \gt(\square_1,\square_2)$
        and $c_1 = \ssym(\square_1)$.
        So by Prop.~\ref{prop:correct-set},
        $(p_1,\qe)$ and
        $(q_1,\pt{\gt(x,y)})$ are correct \wrt{} $P$
        where
        \begin{align*}
            p_1 &= \patterm{\gt(x,y)}
            {\{x\mapsto \ssym(x),y\mapsto \ssym(y)\}}
            {\{x\mapsto\ssym(x),y\mapsto \zero\}} \\
            q_1 &= \patterm{\gt(\ssym(x),\ssym(y))}
            {\{x\mapsto \ssym(x),y\mapsto \ssym(y)\}}
            {\idsub}
        \end{align*}
        Let $n \in \nat$. Then, we have
        $(p_1(n), \qe(n)) \in \binunf(P)$.
        Hence, by Prop.~\ref{prop:binunf_calls},
        $\qe(n)\eta \in \calls_P(p_1(n))$
        for some $\eta \in \subsset$.
        But $\qe(n)\eta = \emptyseq\eta = \emptyseq$
        so by Def.~\ref{def:rewrite-rel}
        ($\calls_P$) we have
        $\sequence{p_1(n)} \rra_P ^+ \emptyseq$
        where $p_1(n) = \gt(\ssym^{n+1}(x),
        \ssym^n(\zero))$.
        \item % We have
        $r_4 = (c'(t'_1,t'_2,t'_1),\emptyseq)$
        and
        $r_5 = (c'(c'_1(x),c'_2(y),c'_2(z)),c'(x,y,z))$
        for 
        $t'_1 = x$, $t'_2 = \zero$,
        $c' = \add(\square_1,\square_2,\square_3)$,
        $c'_1 = \square_1$ and $c'_2 = \ssym(\square_1)$.
        So, by Prop.~\ref{prop:correct-set},
        the pattern rule $(p_2,\qe)$ is correct \wrt{} $P$
        where
        $p_2 = \patterm{\add(x,y,z)}
        {\{y\mapsto \ssym(y),z\mapsto \ssym(z)\}}
        {\{y\mapsto\zero,z\mapsto x\}}$.
    \end{itemize}
\end{example}

Unification for pattern terms is not considered
by~\cite{emmesEG12}. As we need it in our development
(see Def.~\ref{def:patunf} below), we define it here.
\begin{definition}\label{def:pattern-unif}
    Let $p$ and $q$ be pattern terms and $\theta$
    be a pattern substitution.
    Then, $\theta$ is a \emph{unifier} of $p$
    and $q$ if, for all $n \in \nat$, we have
    $p(n)\theta(n) = q(n)\theta(n)$.
    Moreover, $\theta$ is a \emph{most general unifier}
    (mgu) of $p$ and $q$ if, for all $n \in \nat$,
    $\theta(n) \in \mgu(p(n),q(n))$.
    We let $\mgu(p,q)$ denote the set of all mgu's
    of $p$ and $q$.
\end{definition}
\uppereg
if $p = \patterm{\fsym(x,y)}
{\{x \mapsto \ssym(x)\}}{\{x \mapsto \zero\}}$
and $q = \patterm{\fsym(x,y)}
{\{y \mapsto \ssym(y)\}}{\{y \mapsto \one\}}$,
then $\theta = \patsub{\{x \mapsto \ssym(x),
y \mapsto \ssym(y)\}}
{\{x \mapsto \zero,y \mapsto \one\}}$
is a unifier of $p$ and $q$. Indeed, for all
$n \in \nat$, 
$\theta(n) = \{x \mapsto \ssym^n(\zero),
y \mapsto \ssym^n(\one)\}$ is a unifier of
$p(n) = \fsym(\ssym^n(\zero),y)$ and
$p(n) = \fsym(x,\ssym^n(\one))$.
All these notions are naturally extended to
finite sequences of pattern terms.

We also need to adapt the notion of
equivalence class modulo renaming
(see Sect.~\ref{sect:lp}).
% , we have defined the
% equivalence class of a rule modulo renaming.
% We adapt this concept to all pattern rules
% $r$: $[r]$ consists of all pattern rules $r'$
% that describe a subset of $[\rules(r)]$, \ie
% all renamings of the rules described by $r$.
% More formally:
%
\begin{definition}\label{def:pat_equiv_class}
    For all $U \subseteq \patruleset$, we let
    $[U] = \bigcup_{r \in U} [r]$ where
    $[r] =
    \big\{r' \in \patruleset \;\big\vert\;
    rules(r') \subseteq [\rules(r)]
    \big\}$.
\end{definition}
Hence, $[r]$ consists of all pattern rules $r'$
that describe a subset of $[\rules(r)]$.
For instance, if $p = \patterm{\fsym(x,y)}
{\{x \mapsto \ssym(x)\}}{\{x \mapsto \zero\}}$
and $p' = \patterm{\fsym(\ssym(x),y')}
{\{x \mapsto \ssym(x)\}}{\{x \mapsto \zero\}}$,
then $(p',\qe) \in [(p,\qe)]$. Indeed, we have
$\rules((p',\qe)) =
\{(\fsym(\ssym^{n+1}(\zero),y'),\emptyseq)
\mid n \in \nat\} \subseteq
[(\fsym(\ssym^n(\zero),y),\emptyseq)
\mid n \in \nat]
= [\rules((p,\qe))]$.

Now we provide a counterpart of
Def.~\ref{def:binunf} (binary unfolding)
for pattern rules. A notable difference,
however, is the use of an arbitrary set
$B$ instead of 
$E = \left\{\left(\pt{u},\qe\right)
\;\middle\vert\; (u,\emptyseq) \in P
\right\}$.
The set $B$ plays a similar role to the 
pattern creation inference rules
of~\cite{emmesEG12}.
Using suitable sets $B$'s (as those consisting
of rules provided by Prop.~\ref{prop:correct-set}),
we get an approach that computes pattern
rules finitely describing infinite subsets of
$\binunf(P)$, \ie an approach that unfolds
``faster'' (see Ex.~\ref{ex:running1}).
We let $\patid$ denote the set of all pairs
$(\pt{\fsym(x_1,\dots,x_m)},\pt{\fsym(x_1,\dots,x_m)})$
where $\fsym \in \Sigma^{(m)}$ and
$x_1,\dots,x_m$ are distinct variables.
For all rules $r$, the notation $\ll_r$ is
naturally extended to sets of pattern rules, according
to the following definitions: the set of variables
of a pattern term $p = \patterm{s}{\sigma}{\mu}$ is 
$\vars(p) = \vars(s) \cup \vars(\sigma) \cup \vars(\mu)$
and that of a pattern rule
$r = (p,q)$ is $\vars(r) = \vars(p) \cup \vars(q)$.

\begin{definition}\label{def:patunf}
    For all programs $P$ and all
    $B,U \subseteq \patruleset$, we let
    \[T_{P,B}^{\pi}(U) = [B] \cup
    \left[\begin{array}{l}
        (\patterm{u}{\sigma}{\mu}, \\
        \ \patterm{v}{\sigma_i\sigma}{\mu_i\mu} )
    \end{array}
    \; \middle| \;
    \begin{array}{l}
        r = (u,\sequence{v_1,\dots,v_m}) \in P,\
        1 \leq i \leq m \\
        \sequence{(p_1,\qe),
        \dots, (p_{i-1},\qe),
        (p_i,\patterm{v}{\sigma_i}{\mu_i})} \\
        \quad\qquad \ll_r U \cup \patid \\
        \text{if $i < m$ then $v \neq \emptyseq$} \\
        \patsub{\sigma}{\mu} \in
        \mgu\left(\sequence{p_1,\dots,p_i}, 
        \sequence{\pt{v_1},\dots,\pt{v_i}}\right)\\
        \text{$\sigma$ commutes with $\sigma_i$
        and $\mu_i$}
    \end{array}\right]\]
    The \emph{pattern unfolding of $P$ using $B$}
    is the set 
    $\patunf(P,B) =
    \big(T_{P,B}^{\pi}\big)^*(\emptyset)$.
\end{definition}

\begin{example}\label{ex:running1}
    Let $B$ be the set consisting of the
    rules $(p_1,\qe)$ and $(p_2,\qe)$
    of Ex.~\ref{ex:correct}.
    Let us compute some elements of
    $\patunf(P,B)$ by applying
    Def.~\ref{def:patunf}.
    We have $\{(p'_1,\qe), (p'_2,\qe)\}
    \subseteq [B] \subseteq
    T_{P,B}^{\pi}(\emptyset)$ where $p'_1$
    and $p'_2$ are renamed versions of $p_1$
    and $p_2$ respectively:
    \begin{align*}
        p'_1 &= \patterm{\gt(x_1,y_1)}
        {\left\{x_1 \mapsto \ssym(x_1),
        y_1 \mapsto \ssym(y_1)\right\}}
        {\left\{x_1 \mapsto \ssym(x_1),
        y_1 \mapsto \zero\right\}} \\
        p'_2 &= \patterm{\add(x_2,y_2,z_2)}
        {\left\{
        y_2 \mapsto \ssym(y_2),
        z_2 \mapsto \ssym(z_2)\right\}}
        {\left\{
        y_2 \mapsto \zero,
        z_2 \mapsto x_2\right\}}
    \end{align*}
    Let us unfold the whole right-hand side
    of $r_1 \in P$, \ie let us consider $i = m = 3$.
    We have
    $\sequence{(p'_1,\qe),(p'_2,\qe),(p'_3,p'_3)} \ll_{r_1}
    T_{P,B}^{\pi}(\emptyset) \cup \patid$ where
    $p'_3 = \pt{\while(x_3,y_3)}$.
    The right-hand side of $r_1$ is 
    $\sequence{v_1, v_2, v_3} =
    \sequence{\gt(x, y), \add(x, y, z),
    \while(z, \ssym(y))}$. Let 
    $S = \sequence{\pt{v_1},\pt{v_2},\pt{v_3}}$
    and $S' = \sequence{p'_1,p'_2,p'_3}$.
    We show in Ex.~\ref{ex:algo-unif-1}
    that $\patsub{\rho}{\nu} \in \mgu(S,S')$
    where
    \begin{align*}
        \rho &= \big\{
        x \mapsto \ssym(x), \
        y \mapsto \ssym(y),\
        z \mapsto \ssym^2(z),\
        x_2 \mapsto \ssym(x_2),\ 
        x_3 \mapsto \ssym^2(x_3),\
        y_3 \mapsto \ssym(y_3)\big\} \\
        \nu &= \big\{
        x \mapsto \ssym(x_1), \
        y \mapsto \zero,\
        z \mapsto \ssym(x_1),\
        x_2 \mapsto \ssym(x_1),\ 
        x_3 \mapsto \ssym(x_1),\
        y_3 \mapsto \ssym(\zero)\big\}
    \end{align*}
    So, $r'' = \left(\patterm{u}{\rho}{\nu},
    \patterm{\while(x_3,y_3)}{\rho}{\nu}\right)
    \in (T_{P,B}^{\pi})^2(\emptyset)$ where
    $u = \while(x, y)$ is the left-hand side
    of $r_1$.
    It describes the set of binary rules
    \[\left\{r''_n =
    \left(\while(\ssym^{n+1}(x_1),\ssym^n(\zero)),
    \while(\ssym^{2n+1}(x_1),\ssym^{n+1}(\zero))
    \right)
    \;\middle\vert\; n \in \nat\right\}\]
    and we have
    $[r''_n \mid n \in \nat] =
    [r'_n \mid n \in \nat]$
    (see Ex.~\ref{ex:binunf}).
\end{example}

The following result corresponds to the
Soundness Thm.~7 of~\cite{emmesEG12}.

\begin{theorem}\label{theo:soundness}
    Let $P$ be a program and
    $B \subseteq \patruleset$ be
    correct \wrt{} $P$.
    Then, $\patunf(P,B)$ is correct
    \wrt{} $P$.
\end{theorem}

Finally, we adapt the non-termination
criterion of~\cite{emmesEG12} to our
setting.
\begin{theorem}
    \label{theo:detection-nonterm}
    Let $P$ be a program and
    $B \subseteq \patruleset$ be
    correct \wrt{} $P$. Suppose that
    $\patunf(P,B)$ contains a pattern
    rule of the form
    $(\patterm{u}{\sigma}{\mu},
    \patterm{u\sigma^a}{\sigma^b\sigma'}{\mu\mu'})$
    where $\sigma'$ commutes with $\sigma$
    and $\mu$. Then, for all $n \in \nat$
    and all $\theta \in \subsset$, there
    is an infinite $\rra_P$-chain that starts
    from $\sequence{u\sigma^n\mu\theta}$.
\end{theorem}

\begin{example}\label{ex:running2}
    Let us regard the pattern rule $(p,q)$ of
    Ex.~\ref{ex:running}. As
    $rules((p,q)) \subseteq
    \left\{r'_n \mid n \in \nat\right\}$
    with $\left\{r'_n \mid n \in \nat\right\}
    \subseteq \left[r'_n \mid n \in \nat\right]
    = \left[r''_n \mid n \in \nat\right]
    = [\rules(r'')]$ (see Ex.~\ref{ex:running1}),
    we have $(p,q) \in [r'']
    \subseteq \patunf(P,B)$.
    Moreover, $B$ is correct \wrt{}
    $P$ (see Ex.~\ref{ex:correct}) and
    $(p,q) = 
    (\patterm{u\sigma}{\sigma}{\mu},
    \patterm{(u\sigma)\sigma}{\sigma\sigma'}{\mu})$
    (see Ex.~\ref{ex:running})
    where $\sigma'$ commutes with $\sigma$
    and $\mu$.
    By Thm.~\ref{theo:detection-nonterm},
    for all $m,n \in \nat$ and
    $\theta = \{x \mapsto \ssym^n(\zero)\}$,
    the sequence
    $\sequence{(u\sigma)\sigma^m\mu\theta}$
    starts an infinite $\rra_P$-chain, with
    \[\sequence{(u\sigma)\sigma^m\mu\theta} =
    \sequence{u\sigma^{m+1}\mu\theta} =
    \sequence{\while\left(\ssym^{(n+1)+(m+1)}(\zero),
    \ssym^{m+1}(\zero)\right)}\]
    Hence, for all $n > m > 0$, the sequence
    $\sequence{\while(\ssym^n(\zero),
    \ssym^m(\zero))}$
    starts an infinite $\rra_P$-chain.
    This had already been observed 
    in Ex.~\ref{ex:binunf2}.
\end{example}

%%%%%%%%%%%%%%%%%%%%%%%%%%%%%%%%%
\section{Simple Patterns}
\label{sect:simple-patterns}
%%%%%%%%%%%%%%%%%%%%%%%%%%%%%%%%%

In practice, to implement the approach presented
in the previous section, one has to  find a way
to compute mgu's of pattern terms and to check
the non-termination condition of
Thm.~\ref{theo:detection-nonterm}. In this section,
we introduce a class of pattern terms of a special
form, called \emph{simple}, that is more restrictive
but for which we provide a unification algorithm as
well as a non-termination criterion that is easier
to check than that of Thm.~\ref{theo:detection-nonterm}.
We describe them using a new signature that consists
of unary symbols only:
\[\Upsilon = \left\{
c^{a,b} : \text{a unary symbol}
\;\middle\vert\;
c \in \scontext, (a,b) \in \nat^2\right\}\]
Any symbol $c^{a,b} \in \Upsilon$ represents all
successive embeddings of $c$ into itself of the
form $c^{a \times n + b}$ where $n \in \nat$.
Hence, for all $u \in \stermset$ and all
$n \in \nat$, we let $u(n)$ be the element of
$\termset$ obtained from $u$ by replacing every
$c^{a,b} \in \Upsilon$ by  $c^{a \times n + b}$.
Moreover, for all $\theta \in \ssubsset$ and all
$n \in \nat$, we let $\theta(n)$ be the element
of $\subsset$ defined as: for all $x \in X$, 
$(\theta(n))(x) = (\theta(x))(n)$.
In the rest of this section, we consider elements
of $\stermset$ modulo the following equivalence
relation.
\begin{definition}
    The binary relation $\equi \subseteq \stermset^2$
    is defined as: $u \equi v$ iff $u(n) = v(n)$ for
    all $n \in \nat$. It is an equivalence relation
    and we let $[u]$ denote the equivalence class of
    $u$ \wrt{} $\equi$.
\end{definition}

The following straightforward result can be
used to simplify $(\Sigma \cup \Upsilon)$-terms.
\begin{lemma}\label{lem:simple-trick}
    For all $c \in \scontext$,
    $a,b,a',b' \in \nat$ and
    $u \in \stermset$ we have
    $c^{a,b}(c^{a',b'}(u)) \equi
    c^{a+a',b+b'}(u)$
    and $c(u) \equi c^{0,1}(u)$.
\end{lemma}

\begin{example}
    Let $c = \fsym(\square_1, \zero, \square_1)\in \scontext$.
    Then, $c^{1,1} \in \Upsilon$.
    Let $u = c^{1,1}(\one) \in \stermset$.
    For all $n \in \nat$, we have
    $u(n) = c^{n+1}(\one)$.
    For instance, $u(1) = c^2(\one) =
    \fsym(\fsym(\one, \zero, \one), \zero,
    \fsym(\one, \zero, \one))$.
    We note that $v = c^{1,0}(c(\one))
    \equi c^{1,0}(c^{0,1}(\one)) \equi u$;
    indeed, for all $n \in \nat$,
    $v(n) = c^n(c(\one)) = u(n)$.
    Let $\theta = \{x \mapsto u\} \in \ssubsset$
    and $n \in \nat$.
    We have $(\theta(n))(x) = (\theta(x))(n) = u(n)$
    and, for all $y \in X \setminus \{x\}$,
    $(\theta(n))(y) = (\theta(y))(n) = y(n) = y$.
    Hence, $\theta(n) = \{x \mapsto u(n)\}$.
\end{example}

\begin{definition}\label{def:simple-pattern-term}
    A pattern term $p = \patterm{s}{\sigma}{\mu}$
    is called \emph{simple} if, for all $x \in \vars(s)$,
    $\sigma(x) = c^a(x)$ and $\mu(x) = c^b(t)$
    for some $c \in \scontext$, $a,b \in \nat$
    and $t \in \termset$.
    Then, we let $\upsilon(p) = [s\theta_p]$ where
    \[\theta_p = \left\{x \mapsto u
    \;\middle\vert\;
    \begin{array}{l}
        x \in \vars(s),\
        \sigma(x) = c^a(x),\
        \mu(x) = c^b(t) \\
        \text{if $\sigma(x) = x$
        then $u = \mu(x)$ else
        $u = c^{a,b}(t)$}
    \end{array}\right\}\]
    A pattern rule $(p,q)$ is called \emph{simple}
    if $p$ and $q$ are simple.
\end{definition}

The next result follows from
Def.~\ref{def:pattern-term} and
Def.~\ref{def:simple-pattern-term}.
\begin{lemma}\label{lem:simple-pattern-term}
    For all simple pattern terms $p$,
    all $u \in \upsilon(p)$ and
    all $n \in \nat$ we have $p(n) = u(n)$.
\end{lemma}

\newpage
\begin{example}
    The pattern term
    $p = \patterm{s}{\sigma}{\mu}
    = \patterm{\fsym(\ssym(x),y)}
    {\big\{x \mapsto \ssym^2(x)\big\}}
    {\big\{x \mapsto \ssym(x_1),
    y \mapsto \zero\big\}}$
    is simple. For $c = \ssym(\square_1)$,
    we have $\sequence{\sigma(x), \mu(x)}
    = \sequence{c^2(x), c(x_1)}$
    and $\sigma(y) = y$. So
    $\upsilon(p) = [s\theta_p]$ where 
    $\theta_p = \left\{ x \mapsto c^{2,1}(x_1),
    y \mapsto \zero \right\}$.
    Moreover, $s\theta_p =
    \fsym(\ssym(c^{2,1}(x_1)), \zero) =
    \fsym(c(c^{2,1}(x_1)), \zero) \equi
    \fsym(c^{2,2}(x_1), \zero)$.
    For all $n \in \nat$,
    $p(n) = s\sigma^n\mu
    = \fsym(\ssym^{2n+2}(x_1),\zero)
    = (s\theta_p)(n)$. 
\end{example}

We note that the pattern rules
$(p,\qe)$ and $(q,\pt{v})$ produced from
Prop.~\ref{prop:correct-set} are simple.
In Ex.~\ref{ex:correct},
$\upsilon(p_1) =
[\gt(c_1^{1,1}(x), c_1^{1,0}(\zero))]$
and
$\upsilon(p_2) =
[\add(x, c_1^{1,0}(\zero), c_1^{1,0}(x))]$
where $c_1 = \ssym(\square_1)$.

\begin{example}\label{ex:non-simple}
    We illustrate the fact that non-termination
    detection with simple pattern terms is more
    restrictive than with the full class of 
    pattern terms.
    Let $P$ be the program consisting of 
    \begin{align*}
        r_1 &= \left(\while(x,y),
        \sequence{\islist(y), \while(x,\cons(x,y))}\right) \\
        r_2 & = \left(\islist(\nil),\emptyseq\right) \\
        r_3 &= \left(\islist(\cons(x,y)),\islist(y)\right)
    \end{align*}
    Let $c = \cons(x,\square_1)$.
    Then,
    $\sequence{\while(x,c^n(\nil))}
    \mathop{(\rra_{r_1} \circ \rra^n_{r_3} \circ \rra_{r_2})}
    \sequence{\while(x,c^{n+1}(\nil))}$ holds
    for all $n \in \nat$.
    Hence, for all $n \in \nat$,
    $\sequence{\while(x,c^n(\nil))}$ starts
    an infinite $\rra_P$-chain.
    To detect that, one would need an unfolded pattern
    rule of the form 
    $(\patterm{\while(x,y)}{\sigma}{\mu},
    \patterm{\while(x,y)\sigma}{\sigma}{\mu})$
    where $\sigma = \{y \mapsto c(y)\}$
    and $\mu = \{y \mapsto \nil\}$. Such a rule 
    satisfies the condition of
    Thm.~\ref{theo:detection-nonterm} but
    is not simple because $c \not\in \scontext$
    ($\vars(c) \neq \emptyset$). % contains the variable $x$).
    %
    % We note that Prop.~\ref{prop:correct-set}
    % does not produce anything from $(r_2,r_3)$
    % because $r_3$ does not have the required form:
    % its left-hand side is $\islist(c(y))$ where
    % $c = \cons(x,\square_1)$ is not an element of
    % $\scontext$ (because it contains the variable $x$).
    % We also note that for all $n\in\nat$, we have
    % $(\islist(c^n(\nil)), \emptyseq) \in \binunf(P)$,
    % hence $(p, \qe)$ is correct \wrt{} $P$ where
    % $p =
    % \patterm{\islist(y)}{\{y\mapsto c(y)\}}{\{y\mapsto\nil\}}$
    % is not a simple pattern term.
\end{example}

We also define simple substitutions.
\begin{definition}\label{def:simple-pattern-subs}
    A substitution $\theta \in \ssubsset$ is called
    \emph{simple} if, for all $x \in X$,
    $\theta(x) \in [c^{a,b}(t)]$
    for some $c \in \scontext$, $a,b \in \nat$
    and $t \in \termset$.
    Then, we let $\upsilon^{-1}(\theta)$ denote the
    pattern substitution $\patsub{\sigma}{\mu}$
    such that: for all $x \in X$, if
    $\theta(x) \in [c^{a,b}(t)]$ then
    $\sigma(x) = c^a(x)$ and $\mu(x) = c^b(t)$.
\end{definition}

The next result follows from
Def.~\ref{def:pattern-subs} and
Def.~\ref{def:simple-pattern-subs}.
\begin{lemma}\label{lem:simple-pattern-subs}
    For all simple substitutions $\theta$ and
    all $n \in \nat$,
    $\theta(n) = (\upsilon^{-1}(\theta))(n)$.
\end{lemma}

\begin{example}
    Let $\theta = \big\{x \mapsto c^2(\one),
    y \mapsto c(c^{2,1}(c^{1,2}(c(\zero))))
    \big\}$ where $c = \ssym(\square_1)$. 
    We have $\theta(x) \in [c^{0,2}(\one)]$,
    $\theta(y) \in [c^{3,5}(\zero)]$ and
    $\theta(z) \in [c^{0,0}(z)]$ for all
    $z \in X \setminus \{x,y\}$. So,
    $\upsilon^{-1}(\theta) =
    \patsub{\big\{y \mapsto c^3(y)\big\}}
    {\big\{x \mapsto c^2(\one),
    y \mapsto c^5(\zero)\big\}}$.
    For all $n \in \nat$, we have
    $\theta(n) = \big\{x \mapsto c^2(\one),
    y \mapsto c(c^{2n+1}(c^{n+2}(c(\zero))))\big\}
    = \big\{x \mapsto c^2(\one),
    y \mapsto c^{3n + 5}(\zero)\big\}
    = (\upsilon^{-1}(\theta))(n)$.
\end{example}

\subsection{Unification of Simple Pattern Terms}
\label{sect:pattern-unif}
%%%%%%%%%%%%%%%%%%%%%%%%%%%%%%%%%
For all sequences $S = \sequence{p_1,\dots,p_m}$
of simple pattern terms, we define
$\upsilon(S) = \left\{\sequence{u_1,\dots,u_m}
\mid u_1 \in \upsilon(p_1), \dots,
u_m \in \upsilon(p_m) \right\}$.
\begin{algunif}
    Let $S$ and $S'$ be sequences of simple
    pattern terms.
    Let $S_1 \in \upsilon(S)$ and $S'_1 \in \upsilon(S')$.
    \begin{itemize}
        \item If $\mgu(S_1,S'_1)$
        contains a simple substitution $\theta$
        then return $\upsilon^{-1}(\theta)$
        \item else halt with failure.
    \end{itemize}
\end{algunif}

Partial correctness follows from
the next theorem.
\begin{theorem}\label{theo:algo-unif}
    If the unification algorithm successfully
    terminates then it produces a pattern
    substitution which is an mgu of the
    input sequences.
\end{theorem}

In practice, as $S_1$ and $S'_1$ are sequences
of elements of $\stermset$, one can use any
classical unification algorithm
(Robinson, Martelli-Montanari\dots) to 
compute $\theta \in \mgu(S_1,S'_1)$.
Then, it suffices to check whether $\theta$
is simple, for instance using
Lem.~\ref{lem:simple-trick}.

\begin{example}[Related to Ex.~\ref{ex:running1}]
    \label{ex:algo-unif-1}
    Consider the sequences of simple pattern terms
    $S = \sequence{\pt{v_1},\pt{v_2},\pt{v_3}}$ and
    $S' = \sequence{p'_1,p'_2,p'_3}$
    where $v_1 = \gt(x,y)$,
    $v_2 = \add(x,y,z)$,
    $v_3 = \while(z, \ssym(y))$ and
    \begin{align*}
        p'_1 &=
        \patterm{\gt(x_1,y_1)}
        {\left\{x_1 \mapsto \ssym(x_1),
        y_1 \mapsto \ssym(y_1)\right\}}
        {\left\{x_1 \mapsto \ssym(x_1),
        y_1 \mapsto \zero\right\}} \\
        p'_2 &= 
        \patterm{\add(x_2,y_2,z_2)}
        {\left\{
        y_2 \mapsto \ssym(y_2),
        z_2 \mapsto \ssym(z_2)\right\}}
        {\left\{
        y_2 \mapsto \zero,
        z_2 \mapsto x_2\right\}} \\
        p'_3 &= \pt{\while(x_3, y_3)}
    \end{align*}
    Let $c = \ssym(\square_1)$,
    $S_1 =
    \sequence{\gt(x,y), \add(x,y,z), \while(z, c(y))}$ and
    \[S'_1 = \bigsequence{
        \gt\big(c^{1,1}(x_1),c^{1,0}(\zero)\big),
        \add\big(x_2,c^{1,0}(\zero),c^{1,0}(x_2)\big),
        \while(x_3, y_3)}
    \]
    We have $S_1 \in \upsilon(S)$ and
    $S'_1 \in \upsilon(S')$. Moreover,
    $\theta \in \mgu(S_1,S'_1)$ where
    \begin{align*}
        \theta = \big\{ & 
        x \mapsto c^{1,1}(x_1),\
        y \mapsto c^{1,0}(\zero),\
        z \mapsto c^{1,0}(c^{1,1}(x_1)), \\
        & x_2 \mapsto c^{1,1}(x_1),\
        x_3 \mapsto c^{1,0}(c^{1,1}(x_1)),\
        y_3 \mapsto c(c^{1,0}(\zero))\big\}
    \end{align*}
    We note that $\theta(x) \in [c^{1,1}(x_1)]$,
    $\theta(y) \in [c^{1,0}(\zero)]$,
    $\theta(z) \in [c^{2,1}(x_1)]$,
    $\theta(x_2) \in [c^{1,1}(x_1)]$,
    $\theta(x_3) \in [c^{2,1}(x_1)]$ and
    $\theta(y_3) \in [c^{1,1}(\zero)]$.
    So, $\theta$ is a simple substitution and
    the algorithm produces the pattern substitution
    $\upsilon^{-1}(\theta) = \patsub{\rho}{\nu}$ where
    \begin{align*}
        \rho &= \big\{
        x \mapsto \ssym(x), \
        y \mapsto \ssym(y),\
        z \mapsto \ssym^2(z),\
        x_2 \mapsto \ssym(x_2),\ 
        x_3 \mapsto \ssym^2(x_3),\
        y_3 \mapsto \ssym(y_3)\big\} \\
        \nu &= \big\{
        x \mapsto \ssym(x_1), \
        y \mapsto \zero,\
        z \mapsto \ssym(x_1),\
        x_2 \mapsto \ssym(x_1),\ 
        x_3 \mapsto \ssym(x_1),\
        y_3 \mapsto \ssym(\zero)\big\}
    \end{align*}
    By Thm.~\ref{theo:algo-unif},
    $\patsub{\rho}{\nu} \in \mgu(S,S')$.
\end{example}

A natural choice for $S_1$ and $S'_1$ in our
unification algorithm is to consider,
for all
$p  = \patterm{s}{\sigma}{\mu}$ in $S \cup S'$,
the term $s\theta_p \in [s\theta_p]$
(see Def.~\ref{def:simple-pattern-term}).
But this leads to an incomplete approach, \ie
an approach that may fail to find a unifier
even if one exists.

%\newpage
\begin{example}
    Let 
    $p = \patterm{s}{\{x \mapsto c(x)\}}{\idsub}$
    and
    $q = \patterm{s}{\{x \mapsto c_2(x)\}}{\{x \mapsto y\}}$
    where $s = \fsym(x)$,
    $c = \ssym(\square_1)$ and
    $c_2 = c^2$.
    Then, $p$ and $q$ are simple.
    Let $\theta = \patsub{\{x \mapsto c(x)\}}{\{x \mapsto y\}}$.
    For all $n \in \nat$, we have $p(n) = \fsym(c^n(x))$
    and $q(n) = \fsym(c^{2n}(y))$, hence 
    $\theta(n) = \{x \mapsto c^n(y)\}$ is
    a unifier of $p(n)$ and $q(n)$.
    Therefore, $\theta$ is a unifier of
    $p$ and $q$.
    On the other hand, we have
    $s\theta_p = \fsym(c^{1,0}(x))$ and
    $s\theta_q = \fsym(c_2^{1,0}(y))$.
    As $c^{1,0}$ and $c_2^{1,0}$ are different
    symbols, 
    $\mgu(s\theta_p, s\theta_q) = \emptyset$.
    So, from $s\theta_p$ and $s\theta_q$, the
    unification algorithm fails to find a unifier
    for $p$ and $q$.
    Now, let us choose the term
    $u = \fsym(c^{1,0}(c^{1,0}(y)))$ in
    $[s\theta_q]$. The substitution
    $\eta = \{x \mapsto c^{1,0}(y)\}$
    is simple and belongs to $\mgu(s\theta_p, u)$.
    So, the unification algorithm succeeds and
    returns $\upsilon^{-1}(\eta) = \theta$.
\end{example}

\subsection{A Non-Termination Criterion}
\label{sect:nonterm-criteria}
%%%%%%%%%%%%%%%%%%%%%%%%%%%%%%%%%
 
Now, we provide a non-termination criterion
that is simpler to implement than that of
Thm.~\ref{theo:detection-nonterm}. It relies
on pattern rules of the following form, which
is easy to check in practice.

\begin{definition}\label{def:special-pattern-rule}
    We say that a pattern rule $r = (p,q)$ is 
    \emph{special} if it is simple and there
    exists
    \begin{align*}
        c\big(c_1^{a_1,b_1}(t_1),\dots,c_m^{a_m,b_m}(t_m)\big)
        \in \upsilon(p) \quad\text{and}\quad
        c\big(c_1^{a'_1,b'_1}(t_1\rho),\dots,c_m^{a'_m,b'_m}(t_m\rho)\big)
        \in \upsilon(q)
    \end{align*}
    such that $c$ is an $m$-context with $\vars(c) = \emptyset$,
    $\rho \in \subsset$ and
    \begin{enumerate}
        \item \label{theo:detection-nonterm-simple-1}
        $\forall i:
        (t_i \in X) \lor (t_i \in \termset \land \vars(t_i) = \emptyset)$,
        \item \label{theo:detection-nonterm-simple-2}
        $\forall i,j:
        (t_i \in X \land t_i = t_j) \Rightarrow c_i = c_j$,
        \item \label{theo:detection-nonterm-simple-3}
        $\{(a_i,a'_i) \mid \vars(t_i) = \emptyset\} = \{(e,e)\}$
        with $0 < e$, 
        $\{(a_i,a'_i) \mid t_i \in X\} = \{(a,a')\}$
        with $a \leq a'$,
        \item \label{theo:detection-nonterm-simple-4}
        $\{(b_i,b'_i) \mid \vars(t_i) = \emptyset\} = \{(b,b')\}$
        with $b \leq b'$,
        $\{(b_i,b'_i) \mid t_i \in X\} = \{(d,d')\}$,
        \item \label{theo:detection-nonterm-simple-5}
        $k = (b' - b) / e \in \nat$ and
        $a = a' \Rightarrow 0 \leq (d' - d) - a \times k$.
    \end{enumerate}
    Then, we let $\alpha(r) = 0$ if $a = a'$ and
    $\alpha(r) =
    \frac{a \times k - (d' - d)}{a' - a}$ otherwise.
\end{definition}

The existence of a special pattern rule implies
non-termination:
\begin{theorem}
    \label{theo:detection-nonterm-simple}
    Let $P$ be a program and $B \subseteq \patruleset$
    be correct \wrt{} $P$. Suppose that $\patunf(P,B)$
    contains a special pattern rule $r = (p,q)$.
    Then, for all $n \in \nat$ such that
    $n \geq \alpha(r)$ and all $\theta \in \subsset$,
    there is an infinite $\rra_P$-chain that starts
    from $\sequence{p(n)\theta}$.
\end{theorem}

\begin{example}
    In Ex.~\ref{ex:running1}, the set $\patunf(P,B)$
    contains the pattern rule $r'' = (p,q)$ and we
    have
    \begin{align*}
        \while(\ssym^{1,1}(x_1),\ssym^{1,0}(\zero))
        &=
        c\left(c_1^{a_1,b_1}(t_1), c_2^{a_2,b_2}(t_2)\right)
        \in \upsilon(p) \\
        \while(\ssym^{2,1}(x_1),\ssym^{1,1}(\zero))
        &=
        c\left(c_1^{a'_1,b'_1}(t_1), c_2^{a'_2,b'_2}(\zero)\right)
        \in \upsilon(q)
    \end{align*}
    (with a slight abuse of notation when writing
    $\ssym^{1,1}$, $\ssym^{1,0}$ and $\ssym^{2,1}$).
    Moreover, 
    \begin{itemize}
        \item $\{(a_i,a'_i) \mid \vars(t_i) = \emptyset\}
        = \{(a_2,a'_2)\} = \{(1,1)\} = \{(e,e)\}$
        with $0 < e$,
        \item $\{(a_i,a'_i) \mid t_i \in X\}
        = \{(a_1,a'_1)\} = \{(1,2)\} = \{(a,a')\}$
        with $a < a'$,
        \item $\{(b_i,b'_i) \mid \vars(t_i) = \emptyset\}
        = \{(b_2,b'_2)\} = \{(0,1)\} = \{(b,b')\}$ with
        $b \leq b'$,
        \item $\{(b_i,b'_i) \mid t_i \in X\}
        = \{(1,1)\} = \{(d,d')\}$,
        \item $k = (b' - b) / e = (1 - 0) / 1 = 1 \in \nat$.
    \end{itemize}
    So, $\alpha(r) =
    \frac{1 \times 1 - (1 - 1)}{2 - 1} = 1$.
    Then, by Thm.~\ref{theo:detection-nonterm-simple},
    for all $n \in \nat$ such that $n \geq 1$ and all
    $\theta \in \subsset$, there
    is an infinite $\rra_P$-chain that starts from
    $\sequence{p(n)\theta} =
    \sequence{\while(\ssym^{n+1}(x_1),\ssym^n(\zero))\theta}$.
    This corresponds to what we observed in
    Ex.~\ref{ex:binunf2} and Ex.~\ref{ex:running2}.
    For instance, from $n = 1$ and
    $\theta = \{x_1 \mapsto \zero\}$, we get:
    there is an infinite $\rra_P$-chain
    that starts from
    $\sequence{\while(\ssym^2(\zero),\ssym(\zero))}$.
\end{example}

Def.~\ref{def:special-pattern-rule}
requires that 
$A_1 = \{i \mid \vars(t_i) = \emptyset\}$
and
$A_2 = \{i \mid t_i \in X\}$ are not empty.
This can be lifted as follows.
If $A_1 = \emptyset$ or $A_2 = \emptyset$
then we demand that
$a_1 = \dots = a_m = a$,
$b_1 = \dots = b_m = b$,
$a'_1 = \dots = a'_m = a'$ and
$b'_1 = \dots = b'_m = b'$. Moreover:
\begin{itemize}
    \item if $A_1 = \emptyset$ then we replace
    \ref{theo:detection-nonterm-simple-3}--\ref{theo:detection-nonterm-simple-5}
    in Def.~\ref{def:special-pattern-rule} by
    $a \leq a'$ and $a = a' \Rightarrow b \leq b'$;
    we also let $\alpha(r) = 0$ if $a = a'$ and
    $\alpha(r) = \frac{b - b'}{a' - a}$ otherwise;
    \item if $A_2 = \emptyset$ then we demand that
    $a = a'$ and we replace
    \ref{theo:detection-nonterm-simple-3}--\ref{theo:detection-nonterm-simple-5}
    in Def.~\ref{def:special-pattern-rule} by
    $0 < a$ and $k = (b' - b) / a \in \nat$;
    we also let $\alpha(r) = 0$.
\end{itemize}

%%%%%%%%%%%%%%%%%%%%%%%%%%%%%%%%%
\section{Experimental Evaluation}
\label{sect:experiments}
%%%%%%%%%%%%%%%%%%%%%%%%%%%%%%%%%
We have implemented the approach
of Sect.~\ref{sect:simple-patterns}
in our tool \nti{}, which is
the only tool participating in the
\emph{International Termination Competition}%
\footnote{\url{http://termination-portal.org/wiki/Termination_Competition}}
capable of disproving termination
of logic programs (LPs).
We used the natural choice for $S_1$
and $S'_1$ in the unification algorithm,
even if it leads to an incomplete approach
(see end of Sect.~\ref{sect:pattern-unif}).
We ran \nti{} on 41 LPs obtained by translating
term rewrite systems (TRSs) of the
\emph{Termination Problem Data Base}%
\footnote{\url{http://termination-portal.org/wiki/TPDB}}
(TPDB) that are known to be non-looping
non-terminating. These LPs are small but
they are representative of the
kind of non-termination that we want to
capture.
We used the following configuration:
MacBook Pro 2020 with Apple M1 chip,
16~GB RAM, macOS Sequoia 15.4.1.
Table~\ref{table:results-aprove10} shows
the results for 7~LPs obtained from 
directory \verb+AProVE_10+ (which consists of
14~TRSs but we discarded those that translate
to LPs that are not in the scope of our
technique, \ie LPs that terminate or
involve 1-contexts which are not elements
of $\scontext$,
see Ex.~\ref{ex:non-simple}).
Table~\ref{table:results-eeg12} shows the
results for 34~LPs obtained from 
directory \verb+EEG_IJCAR_12+, originally
proposed to evaluate the approach 
of~\cite{emmesEG12} (it consists of 49~TRSs
but, again, we discarded those that translate
to LPs that are out of scope).
The tables have the following structure: column 
``Program'' gives the name of the program together
with its number of rules and relations,
``Mode'' gives the mode of interest
(\verb+i+ means \emph{input}, \ie a 
term with no variable),
``\nti'' gives the non-terminating term
provided by \nti{},
``\#unf'' gives the number of generated
unfolded rules and
``Time(ms)'' gives the time in milliseconds
(we used a time-out of 10 seconds).
The 4~programs marked with $\dagger$
are LP translations of TRSs that have
not been proven non-terminating by any TRS
analyser participating in the competition
until 2024. The results show that our
approach succeeds on them.
On the other hand, our approach fails on 
5~programs.
Our results can be reproduced using our
tool and the benchmarks available at
\url{https://github.com/etiennepayet/nti}.

\begin{table}[p!]
    \centering
    \caption{Logic programs obtained from
    \tt{}TPDB/TRS\_Standard/AProVE\_10}
    \label{table:results-aprove10}
    {\tablefont\begin{tabular}{@{\extracolsep{\fill}}lllrr}
        \topline
        Program (\#rules, \#rel) & Mode & \nti & \#unf & Time(ms)
        \midline
        \verb+andIsNat+ (3, 2) & \verb+f(i,i)+ & \verb+f(0,0)+ & 5 & 100 \\
        % \verb+challenge_fab+ (6, 4) & \verb+from(i)+ & ? & & \\
        \verb+double+ (3, 2) & \verb+f(i)+ & \verb+f(0)+ & 4 & 100 \\
        % \verb+downfrom+ (5, 3) & \verb+f(i)+ & ? & & \\
        \verb+ex1+ (3, 2) & \verb+f(i,i)+ & \verb+f(0,0)+ & 4 & 100 \\
        \verb+ex2+ (3, 2) & \verb+g(i)+ & \verb+g(0)+ & 4 & 100 \\
        \verb+ex3+ (4, 2) & \verb+g(i,i)+ & \verb+g(0,0)+ & 12 & 100 \\
        % \verb+ex4+ (5, 3) & \verb+add(i,i)+ & ? & & \\
        \verb+halfdouble+ (7, 4) & \verb+f(i)+ & \verb+f(0)+ & 10 & 100 \\
        % \verb+isList+ (3, 2) & \verb+f(nil)+ & \verb+f(0)+ & 6 & 90 \\ % this one loops
        \verb+isNat+ (3, 2) & \verb+f(i)+ & \verb+f(0)+ & 4 & 100
        \botline
    \end{tabular}}
\end{table}

\begin{table}[p!]
    \centering
    \caption{Logic programs obtained from
    \tt{}TPDB/TRS\_Standard/EEG\_IJCAR\_12}
    \label{table:results-eeg12}
    {\tablefont\begin{tabular}{@{\extracolsep{\fill}}lllrr}
        \topline
        Program (\#rules, \#rel) & Mode & \nti & \#unf & Time(ms)
        \midline
        \verb+emmes-nonloop-ex1_1+ (5, 4) & \verb+f(i,i)+ & \verb+f(s(0),0)+ & 6  & 90  \\
        \verb+emmes-nonloop-ex1_2+ (7, 5) $\dagger$ & \verb+f(i,i)+ & \verb+f(s(0),0)+ & 38 & 120 \\
        \verb+emmes-nonloop-ex1_3+ (7, 5) $\dagger$ & \verb+f(i,i)+ & \verb+f(s(0),0)+ & 38 & 120 \\
        \verb+emmes-nonloop-ex1_4+ (7, 5) & \verb+f(i,i)+ & \verb+f(s(0),0)+ & 29 & 120 \\
        \verb+emmes-nonloop-ex1_5+ (7, 5) & \verb+f(i,i)+ & \verb+f(s(0),0)+ & 29 & 130 \\
        \verb+emmes-nonloop-ex2_1+ (6, 4) $\dagger$ & \verb+f(i,i)+ & \verb+f(s(0),0)+ & 26 & 120 \\
        \verb+emmes-nonloop-ex2_2+ (5, 3) & \verb+f(i,i)+ & \verb+f(s(0),0)+ & 8  & 100 \\
        \verb+emmes-nonloop-ex2_3+ (6, 4) $\dagger$ & \verb+f(i,i)+ & \verb+f(s(0),0)+ & 26 & 120 \\
        \verb+emmes-nonloop-ex2_4+ (8, 5) & \verb+f(i,i)+ & \verb+f(s(0),0)+ & 140 & 240 \\
        \verb+emmes-nonloop-ex2_5+ (8, 5) & \verb+f(i,i)+ & \verb+f(s(0),0)+ & 140 & 270 \\
        \verb+emmes-nonloop-ex3_1+ (6, 4) & \verb+f(i)+ & \verb+f(s(0))+ & 30  & 120 \\
        \verb+emmes-nonloop-ex3_2+ (6, 4) & \verb+f(i)+ & \verb+f(s(0))+ & 30  & 110 \\
        \verb+emmes-nonloop-ex3_3+ (8, 5) & \verb+f(i)+ & \verb+f(s(0))+ & 78  & 180 \\
        \verb+emmes-nonloop-ex3_4+ (8, 5) & \verb+f(i)+ & \verb+f(s(0))+ & 178 & 270 \\
        \verb+emmes-nonloop-ex4_1+ (7, 4) & \verb+f(i)+ & \verb+f(0)+ & 12 & 100 \\
        \verb+emmes-nonloop-ex4_2+ (9, 5) & \verb+f(i)+ & \verb+f(0)+ & 40 & 140 \\
        \verb+emmes-nonloop-ex4_3+ (9, 5) & \verb+f(i)+ & \verb+f(0)+ & 40 & 130 \\
        \verb+emmes-nonloop-ex4_4+ (9, 5) & \verb+f(i)+ & \verb+f(0)+ & 39 & 140 \\
        \verb+emmes-nonloop-ex5_1+ (9, 5) & \verb+f(i)+ & \verb+f(s(0))+ & 40 & 140 \\
        \verb+emmes-nonloop-ex5_2+ (8, 4) & \verb+f(i)+ & \verb+f(s(0))+ & 14 & 100 \\
        \verb+emmes-nonloop-ex5_3+ (9, 5) & \verb+f(i)+ & \verb+f(s(0))+ & 40 & 140 \\
        \verb+enger-nonloop-ex_payet+ (6, 3) & \verb+while(i,i)+ & ? & 1296 & time out \\
        \verb+enger-nonloop-isDNat+ (3, 2) & \verb+f(i)+ & \verb+f(0)+ & 4 & 90 \\
        \verb+enger-nonloop-isTrueList+ (3, 2) & \verb+f(i)+ & \verb+f(nil)+ & 4 & 90 \\
        \verb+enger-nonloop-swap_decr+ (5, 3) & \verb+f(i)+ & ? & 41017 & time out\\
        \verb+enger-nonloop-swapX+ (3, 2) & \verb+g(i)+ & \verb+g(0)+ & 4 & 90 \\
        \verb+enger-nonloop-swapXY+ (3, 2) & \verb+g(i,i)+ & \verb+g(0,0)+ & 4 & 90 \\
        \verb+enger-nonloop-swapXY2+ (3, 2) & \verb+g(i,i)+ & \verb+g(0,0)+ & 4 & 90 \\
        \verb+enger-nonloop-toOne+ (5, 3) & \verb+f(i)+ & \verb+f(s(0))+ & 7 & 90 \\
        \verb+enger-nonloop-unbounded+ (3, 2) & \verb+h(i,i)+ & \verb+h(s(0),0)+ & 4 & 90 \\
        \verb+enger-nonloop-while-lt+ (3, 2) & \verb+while(i,i)+ & \verb+while(0,0)+ & 4 & 90 \\
        \verb+rybalchenko-nonloop-popl08+ (15, 7) & \verb+while(i,i)+ & ? & 9505 & time out \\
        \verb+velroyen-nonloop-AlternatingIncr_c+ (11, 6) & \verb+while(i)+ & ? & 2035 & time out \\
        \verb+velroyen-nonloop-ConvLower_c+ (12, 6) & \verb+while(i)+ & ? & 1341 & time out
        \botline
    \end{tabular}}
\end{table}

%%%%%%%%%%%%%%%%%%%%%%%
\section{Related Work}\label{sect:rel-work}
%%%%%%%%%%%%%%%%%%%%%%%
The only other approach we are aware of for
proving non-looping non-termination of logic
programs is that of~\cite{payet24}. Roughly,
it detects infinite chains of the form
$\seqset{s}_0
\mathop{(\rra^*_{r_1} \circ \rra_{r_2})}
\seqset{s}_1 \mathop{(\rra^*_{r_1} \circ \rra_{r_2})}
\cdots$ where $(r_1,r_2)$ is a \emph{recurrent pair}
of binary rules (we note that the infinite
$\rra_P$-chain of Ex.~\ref{ex:binunf2} does not
have this form). 
% Our first experiments show that
This approach
seems to address another class of non-loopingness
(compared to that of this paper):
it is not able to disprove termination
of the programs of Tables~\ref{table:results-aprove10}
and~\ref{table:results-eeg12} and, on the other hand,
it is able to disprove termination of programs%
\footnote{\uppereg those in directories
\verb+TPDB/Logic_Programming/Payet_22+,
\verb+TPDB/Logic_Programming/Payet_23+
and \verb+TPDB/Logic_Programming/Payet_24+,
proposed to evaluate the approach
of~\cite{payet24}.}
on which our approach fails.

Loop checking (see, \eg the paper
by~\cite{bolAK91}) is also related to
our work. It attempts to prune infinite rewrites
at runtime using necessary conditions for the
existence of infinite chains (hence, there is
a risk of pruning a finite rewrite). In contrast,
our approach uses a sufficient condition
(see Thm.~\ref{theo:detection-nonterm}
and Thm.~\ref{theo:detection-nonterm-simple})
to \emph{prove} the existence of atomic
goals that start a non-terminating chain.

Tabling (see, \eg the paper
by~\cite{sagonasSW94}) is another related
technique that avoids infinite rewrites by
storing intermediate results in a table and
reusing them when needed. Only loops of a 
particular simple form are detected (\ie
when a variant occurs in the evaluation of a
subgoal). We are not aware of any
tabling-based approach that can capture
non-looping non-termination.

Another related technique is that 
of~\cite{payetM06} which proves the
existence of loops using neutral 
argument positions of predicate symbols.

%%%%%%%%%%%%%%%%%%%%%%%
\section{Conclusion}\label{sect:conclusion}
%%%%%%%%%%%%%%%%%%%%%%%
We have presented a new approach, based on
a new unfolding technique that generates
correct patterns of rules, to disprove
termination of logic programs.
We have implemented it in our tool \nti{}
and we have successfully evaluated it on
logic programs obtained by translating
TRSs from the TPDB.

Future work will be concerned with completeness
of our unification algorithm, extending
Prop.~\ref{prop:correct-set} to get more
initial pattern rules and taking into account
1-contexts with variables in simple pattern
terms (to deal with programs as that of
Ex.~\ref{ex:non-simple}).
We also plan to adapt our approach to TRSs
and to compare it to that
of~\cite{emmesEG12}.
Moreover, we will compare our approach
to that of~\cite{payet24} on a theoretical
level.

%%%%%%%%%%%%%%%%%%%%%%%%%%%%%%%%%%%
% Bibliography
%%%%%%%%%%%%%%%%%%%%%%%%%%%%%%%%%%%
\bibliographystyle{tlplike}
% \bibliography{biblio}

\begin{thebibliography}{}

\bibitem[\protect\citeauthoryear{Apt}{Apt}{1997}]{apt97}
{\sc Apt, K.~R.} 1997.
\newblock {\em From Logic Programming to {P}rolog}.
\newblock Prentice Hall International series in computer science. Prentice
  Hall.

\bibitem[\protect\citeauthoryear{Baader and Nipkow}{Baader and
  Nipkow}{1998}]{baaderN98}
{\sc Baader, F.} {\sc and} {\sc Nipkow, T.} 1998.
\newblock {\em Term Rewriting and All That}.
\newblock Cambridge University Press.

\bibitem[\protect\citeauthoryear{Bol, Apt, and Klop}{Bol
  et~al\mbox{.}}{1991}]{bolAK91}
{\sc Bol, R.~N.}, {\sc Apt, K.~R.}, {\sc and} {\sc Klop, J.~W.} 1991.
\newblock An analysis of loop checking mechanisms for logic programs.
\newblock {\em Theoretical Computer Science\/}~{\em 86,\/}~1, 35--79.

\bibitem[\protect\citeauthoryear{Codish and Taboch}{Codish and
  Taboch}{1999}]{codishT99}
{\sc Codish, M.} {\sc and} {\sc Taboch, C.} 1999.
\newblock {A} semantic basis for the termination analysis of logic programs.
\newblock {\em Journal of Logic Programming\/}~{\em 41,\/}~1, 103--123.

\bibitem[\protect\citeauthoryear{Emmes, Enger, and Giesl}{Emmes
  et~al\mbox{.}}{2012}]{emmesEG12}
{\sc Emmes, F.}, {\sc Enger, T.}, {\sc and} {\sc Giesl, J.} 2012.
\newblock Proving non-looping non-termination automatically.
\newblock In {\em Proc. of the 6th International Joint Conference on Automated
  Reasoning (IJCAR'12)}, {B.~Gramlich}, {D.~Miller}, {and} {U.~Sattler}, Eds.
  LNCS, vol. 7364. Springer, 225--240.

\bibitem[\protect\citeauthoryear{Payet}{Payet}{2024}]{payet24}
{\sc Payet, E.} 2024.
\newblock Non-termination in term rewriting and logic programming.
\newblock {\em Journal of Automated Reasoning\/}~{\em 68,\/}~4, 24~pages.

\bibitem[\protect\citeauthoryear{Payet and Mesnard}{Payet and
  Mesnard}{2006}]{payetM06}
{\sc Payet, E.} {\sc and} {\sc Mesnard, F.} 2006.
\newblock Nontermination inference of logic programs.
\newblock {\em {ACM} Transactions on Programming Languages and Systems\/}~{\em
  28,\/}~2, 256--289.

\bibitem[\protect\citeauthoryear{Sagonas, Swift, and Warren}{Sagonas
  et~al\mbox{.}}{1994}]{sagonasSW94}
{\sc Sagonas, K.}, {\sc Swift, T.}, {\sc and} {\sc Warren, D.~S.} 1994.
\newblock {XSB} as an efficient deductive database engine.
\newblock In {\em Proc. of the 1994 {ACM} {SIGMOD} International Conference on
  Management of Data}, {R.~T. Snodgrass} {and} {M.~Winslett}, Eds. {ACM} Press,
  442--453.

\end{thebibliography}

%%%%%%%%%%%%%%%%%%%%%%%%%%%%%%%%%%%
% Appendix
%%%%%%%%%%%%%%%%%%%%%%%%%%%%%%%%%%%
\newpage
\appendix
% \input{appendix}

%%%%%%%%%%%%%%%%%%%%%%%%
\section{Proofs}
%%%%%%%%%%%%%%%%%%%%%%%%

In this appendix, we provide the proofs of the results
presented in the paper.

\subsection{Proof of Proposition~\ref{prop:binunf_calls}
(Section~\ref{sect:binunf})}
%%%%%%%%%%%%%%%%%%%%%%%%%%%%%%%%%%%%%%%%%%%
Let $\gamma$ be a renaming such that
$(u\gamma,v\gamma)$ is variable disjoint
from $u\sigma$. Let
$\eta = \{x\gamma \mapsto x\sigma
\mid x \in \vars(u),\
x\gamma \neq x\sigma\}$.
\begin{itemize}
    \item First, we prove that $\eta$ is a
    substitution. Let $(x \mapsto s)$ and
    $(y \mapsto t)$ be some bindings in $\eta$.
    By definition of $\eta$, we have
    $(x \mapsto s) = (x'\gamma \mapsto x'\sigma)$
    and
    $(y \mapsto t) = (y'\gamma \mapsto y'\sigma)$
    for some variables $x'$ and $y'$ in
    $\vars(u)$. As $\gamma$ is a variable
    renaming, it is a bijection on $X$, so
    if $x = y$ then $x' = y'$ and hence
    $(x \mapsto s) = (y \mapsto t)$.
    Consequently, for any bindings
    $(x \mapsto s)$ and $(y \mapsto t)$ in
    $\eta$, $(x \mapsto s) \neq (y \mapsto t)$
    implies $x\neq y$. Moreover, by definition
    of $\eta$, for any $(x \mapsto s) \in \eta$
    we have $x \neq s$.
    Therefore, $\eta$ is a substitution.
    \item Then, we prove that $\eta$ is a
    unifier of $u\gamma$ and $u\sigma$.
    \begin{itemize}
        \item Let $x\in\vars(u)$. Then, by
        definition of $\eta$,
        $x\gamma\eta=x\sigma$. So,
        $u\gamma\eta = u\sigma$.
        \item Let $y\in\dom(\eta)$. Then,
        $y = x\gamma$ for some $x \in \vars(u)$.
        Hence, $y \in \vars(u\gamma)$. As $u\gamma$
        is variable disjoint from $u\sigma$, we have
        $y \not\in \vars(u\sigma)$. Therefore, we have
        $\dom(\eta) \cap \vars(u\sigma) = \emptyset$,
        so $u\sigma\eta = u\sigma$.
    \end{itemize}
    Consequently, we have
    $u\gamma\eta = u\sigma\eta$.
\end{itemize}
Hence, $u\gamma$ and $u\sigma$ unify,
so $\mgu(u\gamma, u\sigma) \neq \emptyset$
(see, \eg Thm.~4.5.8 of~\cite{baaderN98}).
Let $\rho \in \mgu(u\gamma, u\sigma)$.
By Prop.~4.2 of~\cite{codishT99},
we have $\calls_P(u\sigma) =
\{v'\rho \mid
(u',v') \in \binunf(P),\
\rho \in \mgu(u\sigma,u')\}$.
Consequently, as
$(u\gamma,v\gamma) \in \binunf(P)$
(because $\binunf(P)$ is closed
by renaming), for $\theta = \gamma\rho$
we have $v\theta \in \calls_P(u\sigma)$.

\subsection{Proof of Proposition~\ref{prop:correct-set}
(Section~\ref{sect:patterns})}
%%%%%%%%%%%%%%%%%%%%%%%%%%%%%%%%%%%%%%%%%%%
We proceed in two steps.

\begin{proof}[Proof that $(p,\qe)$ is correct \wrt{} $P$]
    We have $p = \patterm{v}{\sigma}{\mu}$ where
    \begin{align*}
        v &= c(x_1,\dots,x_m) \quad
        \text{with $\vars(c) = \emptyset$} \\
        \sigma &= \left\{x_k \mapsto c_k(x_k)
        \;\middle\vert\; 1 \leq k \leq m,\
        c_k(x_k) \neq x_k \right\} \quad
        \text{with $\vars(c_1,\dots,c_m) = \emptyset$} \\
        \mu &= \left\{x_k \mapsto t_k
        \;\middle\vert\; 1 \leq k \leq m,\
        t_k \neq x_k\right\}
    \end{align*}
    
    We prove by induction on $n$ that
    $(p(n),\qe(n)) \in \binunf(P)$ for all
    $n \in \nat$, where
    $(p(n),\qe(n)) = (p(n),\emptyseq)$.
    \begin{itemize}
        \item (Base: $n = 0$) Here, 
        $p(n) = v\sigma^n\mu = v\mu
        = c(t_1,\cdots,t_m) = u'$, so 
        $(p(n),\emptyseq) = (u',\emptyseq) = r'$.
        But, by Def.~\ref{def:binunf},
        $r' \in \binunf(P)$. Hence,
        $(p(n),\emptyseq) \in \binunf(P)$.
        \item (Induction)
        Suppose that for some $n \in \nat$
        we have $(p(n),\emptyseq) \in \binunf(P)$.
        Note that
        $p(n) = v\sigma^n\mu = 
        c(c_1^n(t_1),\dots,c_m^n(t_m))$.
        
        Let us apply Def.~\ref{def:binunf} to
        unfold the right-hand side of $r = (u,v)$
        using $(p(n),\emptyseq)$. Let $\gamma$
        be a renaming such that 
        $(p(n)\gamma,\emptyseq\gamma) =
        (p(n)\gamma,\emptyseq)
        $ is variable disjoint from $r$. Then, we have
        $(p(n)\gamma,\emptyseq) \ll_{r} \binunf(P)$.
        Note that $\vars(c) = \emptyset$. Moreover,
        \begin{align*}
            p(n)\gamma &=
            c(c_1^n(t_1),\dots,c_m^n(t_m))\gamma =
            c(c_1^n(t_1\gamma),\dots,c_m^n(t_m\gamma))
            \quad\text{and} \\
            v &= c(x_1,\dots,x_m)
        \end{align*}
        are variable disjoint,
        so $x_1,\dots,x_m$ do not occur in
        $p(n)\gamma$. Consequently, the
        Martelli-Montanari unification
        algorithm (see, \eg Sect.~2.6 of
        the book by~\cite{apt97})
        applied to $p(n)\gamma$ and $v$
        returns the mgu
        $\theta = \left\{x_k \mapsto c_k^n(t_k\gamma)
        \;\middle\vert\; 1 \leq k \leq m\right\}$.
        So, by Def.~\ref{def:binunf}
        $(u\theta,\emptyseq\theta) \in \binunf(P)$
        with $\emptyseq\theta = \emptyseq$ and
        \begin{align*}
            u\theta = c(c_1(x_1),\dots,c_m(x_m))\theta 
            &= c(c_1(x_1\theta),\dots,c_m(x_m\theta)) \\
            &= c(c^{n+1}_1(t_1\gamma),\dots,
            c^{n+1}_m(t_m\gamma)) \\
            &= c(c^{n+1}_1(t_1),\dots,
            c^{n+1}_m(t_m))\gamma \\
            &= (v\sigma^{n+1}\mu)\gamma = p(n+1)\gamma
        \end{align*}
        Hence, $(p(n+1)\gamma,\emptyseq)
        \in \binunf(P)$, \ie
        $(p(n+1),\emptyseq) \in \binunf(P)$
        (because $\binunf(P)$ is closed by renaming).
    \end{itemize}
    So, for all $n \in \nat$,
    $(p(n),\qe(n)) \in \binunf(P)$, \ie
    $(p,\qe)$ is correct \wrt{} $P$.
\end{proof}

\begin{proof}[Proof that $(q,\pt{v})$ is correct \wrt{} $P$]
    We have $q = \patterm{u}{\sigma}{\idsub}$ where
    \begin{align*}
        u &= c(c_1(x_1),\dots,c_m(x_m))
        \text{ with $\vars(c) = 
        \vars(c_1,\dots,c_m) = \emptyset$} \\
        \sigma &= \{x_k \mapsto c_k(x_k)
        \mid 1 \leq k \leq m,\
        c_k(x_k) \neq x_k \}
    \end{align*}

    We prove by induction on $n$ that
    $(q(n),\pt{v}(n)) \in \binunf(P)$ for all
    $n \in \nat$, where
    $(q(n),\pt{v}(n)) = (q(n),v)$.
    \begin{itemize}
        \item (Base: $n = 0$) Here, 
        $q(n) = u\sigma^n = u$, so 
        $(q(n),v) = (u,v) = r$.
        By Def.~\ref{def:binunf},
        $r \in \binunf(P)$ (obtained by 
        unfolding the right-hand side of
        $r$ using $\id$). So,
        $(q(n),v) \in \binunf(P)$.
        \item (Induction)
        Suppose that for some $n \in \nat$
        we have $(q(n),v) \in \binunf(P)$.
        Note that
        \begin{align*}
            q(n) &= u\sigma^n = 
            c(c_1^{n + 1}(x_1),\dots,
            c_m^{n + 1}(x_m))
        \end{align*}
        
        Let us apply Def.~\ref{def:binunf} to
        unfold the right-hand side of $r = (u,v)$
        using $(q(n),v)$. Let $\gamma$
        be a renaming such that 
        $(q(n)\gamma,v\gamma)$ is variable
        disjoint from $r$. Then, we have
        $(q(n)\gamma,v\gamma) \ll_{r} \binunf(P)$.
        Note that $\vars(c) = \emptyset$. Moreover,
        \begin{align*}
            q(n)\gamma &= c(c_1^{n + 1}(x_1\gamma),
            \dots, c_m^{n + 1}(x_m\gamma))
            \quad\text{and} \\
            v &= c(x_1,\dots,x_m)
        \end{align*}
        are variable disjoint, so $x_1,\dots,x_m$ do
        not occur in $q(n)\gamma$.
        So, the Martelli-Montanari unification
        algorithm (see, \eg Sect.~2.6 of the
        book by ~\cite{apt97})
        applied to $q(n)\gamma$ and $v$ returns the mgu
        $\theta = \left\{x_k \mapsto
        c_k^{n + 1}(x_k\gamma)
        \;\middle\vert\; 1 \leq k \leq m
        \right\}$.
        By Def.~\ref{def:binunf}
        $(u\theta,v\gamma\theta) \in \binunf(P)$.
        We have
        \begin{align*}
            u\theta =
            c(c_1(x_1),\dots,c_m(x_m))\theta
            &=
            c(c_1(x_1\theta),\dots,c_m(x_m\theta)) \\
            &= c(c_1^{n + 2}(x_1\gamma),\dots,c_m^{n + 2}(x_m\gamma)) \\
            &= c(c_1^{n + 2}(x_1),\dots,c_m^{n + 2}(x_m))\gamma \\
            &= (u\sigma^{n + 1})\gamma
            = q(n + 1)\gamma
        \end{align*}
        Moreover, $v\gamma\theta = v\gamma$.
        Indeed, $\vars(v\gamma) =
        \{x_1\gamma,\dots,x_m\gamma\}$
        and $\dom(\theta) = \{x_1,\dots,x_m\}$,
        with
        $\{x_1,\dots,x_m\} \cap
        \{x_1\gamma,\dots,x_m\gamma\} = \emptyset$
        (because $x_1,\dots,x_m$ do not occur in
        $q(n)\gamma$).
        Hence, $(u\theta,v\gamma\theta) =
        (q(n+1)\gamma, v\gamma) \in \binunf(P)$,
        \ie
        $(q(n+1), v) \in \binunf(P)$
        (because $\binunf(P)$ is closed by 
        renaming).
    \end{itemize}
    So, for all $n \in \nat$,
    $(q(n), v) \in \binunf(P)$, \ie
    $(q, \pt{v})$ is correct \wrt{} $P$.
\end{proof}

\subsection{Proof of Theorem~\ref{theo:soundness}
(Section~\ref{sect:patterns})}
%%%%%%%%%%%%%%%%%%%%%%%%%%%%%%%%%%%%%%%%%%%

First, we establish the following result.
\begin{lemma}\label{lem:correct_binunf}
    Let $P$ be a program and $r$ be
    a pattern rule. If $r$ is correct
    \wrt{} $P$ then $[r]$ also is.
\end{lemma}
\begin{proof}
    Suppose that $r$ is correct \wrt{} $P$.
    By Def.~\ref{def:correct}, we have to prove
    that every pattern rule in $[r]$ also is.
    Let $r' \in [r]$. Then, 
    by Def.~\ref{def:pat_equiv_class}, we have 
    $rules(r') \subseteq [\rules(r)]$. As $r$
    is correct \wrt{} $P$, by Def.~\ref{def:correct}
    we have $\rules(r) \subseteq \binunf(P)$. But,
    by Def.~\ref{def:binunf}, $\binunf(P)$ is
    closed by renaming, \ie
    $[\rules(r)] \subseteq \binunf(P)$.
    Therefore, we have
    $rules(r') \subseteq \binunf(P)$,
    \ie $r'$ is correct \wrt{} $P$.
\end{proof}

Then, using Lem.~\ref{lem:correct_binunf},
we prove that the unfolding operator
$T_{P,B}^{\pi}$ is sound.
\begin{proposition}\label{prop:soundness}
    Let $P$ be a program and
    $B,U \subseteq \patruleset$.
    If $B \cup U$ is correct \wrt{} $P$
    then $T_{P,B}^{\pi}(U)$ also is.
\end{proposition}
\begin{proof}
    Suppose that $B \cup U$ is correct
    \wrt{} $P$, which implies that $B$
    and $U$ are correct \wrt{} $P$
    (see Def.~\ref{def:correct}).
    Let $r \in T_{P,B}^{\pi}(U)$.
    % We prove that $rules(r) \subseteq \binunf(P)$.
    By Def.~\ref{def:patunf}, we have
    $r \in [r']$ where $r'=(p,q)$ is a pattern
    rule which has two possible origins.
    \begin{itemize}
        \item Either $r' \in B$. Then, as $B$
        is correct \wrt{} $P$, $r'$ also is
        (Def.~\ref{def:correct}).
        \item Or $r'$ is constructed by unfolding
        a prefix of the right-hand side of a rule
        \[r_1 = (u,\sequence{v_1,\dots,v_m}) \in P\]
        Then, $p = \patterm{u}{\sigma}{\mu}$
        and $q = \patterm{v}{\sigma_i\sigma}{\mu_i\mu}$
        where
        \begin{itemize}
            \item $1 \leq i \leq m$,
            \item $\big\langle(p_1,\qe),\dots,(p_{i-1},\qe),
            (p_i,\patterm{v}{\sigma_i}{\mu_i})\big\rangle
            \ll_{r_1} U \cup \patid$,
            \item $\patsub{\sigma}{\mu} \in
            \mgu\left(\sequence{p_1,\dots,p_i},
            \sequence{\pt{v_1},\dots,\pt{v_i}}\right)$,
            \item $\sigma$ commutes with $\sigma_i$ and
            $\mu_i$,
            \item if $i < m$ then $v \neq \emptyseq$.
        \end{itemize}
        Let $n \in \nat$.
        \begin{itemize}
            \item For all $1 \leq j < i$, we
            necessarily have 
            $(p_j,\qe) \in U$
            (because no pattern rule from
            $\patid$ has $\qe$ as right-hand side)
            so, as $U$ is correct \wrt{} $P$, we have
            $(p_j(n),\emptyseq) \in \binunf(P)$.
            \item The rule
            $(p_i,\patterm{v}{\sigma_i}{\mu_i})$
            belongs to $U \cup \patid$.
            If it belongs to $U$ then, as $U$ is correct
            \wrt{} $P$, we have
            $(p_i(n),v\sigma^n_i\mu_i) \in \binunf(P)$.
            If it belongs to $\patid$ then 
            $(p_i(n),v\sigma^n_i\mu_i) \in \id$.
        \end{itemize}
        Moreover, the rules $(p_1(n),\emptyseq),\dots,
        (p_{i-1}(n),\emptyseq),
        (p_i(n),v\sigma^n_i\mu_i)$ and $r_1$
        are variable disjoint from 
        each other. Hence, we have
        \[\sequence{(p_1(n),\emptyseq),\dots,
        (p_{i-1}(n),\emptyseq),
        (p_i(n),v\sigma^n_i\mu_i)}
        \ll_{r_1} \binunf(P) \cup \id\]
        with
        $v\sigma^n_i\mu_i \neq \emptyseq$
        if $i < m$ (because $\sigma_i,\mu_i \in \subsset$).
        On the other hand, as
        $\patsub{\sigma}{\mu} \in
        \mgu\left(\sequence{p_1,\dots,p_i},
        \sequence{\pt{v_1},\dots,\pt{v_i}}\right)$, by
        Def.~\ref{def:pattern-unif} we have
        \[(\patsub{\sigma}{\mu})(n) = \sigma^n\mu \in
        \mgu\left(\sequence{p_1(n),\dots,p_i(n)},
        \sequence{(\pt{v_1})(n),\dots,(\pt{v_i})(n)}\right)\]
        where
        $\sequence{(\pt{v_1})(n),\dots,(\pt{v_i})(n)} =
        \sequence{v_1,\dots,v_i}$.
        So, by Def.~\ref{def:binunf},
        we have $(u\sigma^n\mu,
        (v\sigma^n_i\mu_i)\sigma^n\mu)
        \in \binunf(P)$. But $u\sigma^n\mu = p(n)$
        and, as $\sigma$ commutes with
        $\sigma_i$ and $\mu_i$,
        $(v\sigma^n_i\mu_i)\sigma^n\mu =
        v(\sigma_i\sigma)^n(\mu_i\mu) = q(n)$.
        Therefore, we have proved that
        $(p(n),q(n)) \in \binunf(P)$, \ie
        as $n$ is arbitrary, that
        $\rules(r') \subseteq \binunf(P)$, \ie
        $r'$ is correct \wrt{} $P$
        (Def.~\ref{def:correct}).
    \end{itemize}
    Hence, we have proved that $r'$ is correct
    \wrt{} $P$, whatever its origin.
    So, by Lem.~\ref{lem:correct_binunf},
    $[r']$ is correct \wrt{} $P$.
    As $r \in [r']$, this implies that
    $r$ also is (Def.~\ref{def:correct}).
    
    Finally, as $r$ denotes an arbitrary
    element of $T_{P,B}^{\pi}(U)$, by
    Def.~\ref{def:correct} we have proved
    that $T_{P,B}^{\pi}(U)$ is correct
    \wrt{} $P$.
    %\qed
\end{proof}

Finally, using Prop.~\ref{prop:soundness},
we can now prove
Thm.~\ref{theo:soundness}.
\begin{proof}[Proof of Theorem~\ref{theo:soundness}]
    We prove that for all $n \in \nat$,
    $(T_{P,B}^{\pi})^n(\emptyset)$
    is correct \wrt{} $P$.
    We proceed by induction on $n$.
    \begin{itemize}
        \item (Base: $n = 0$) Here, 
        $(T_{P,B}^{\pi})^n(\emptyset) =
        \emptyset$ and it is vacuously true
        that $\emptyset$ is correct \wrt{} $P$.
        \item (Induction) 
        If, for some $n \in \nat$,
        $(T_{P,B}^{\pi})^n(\emptyset)$
        is correct \wrt{} $P$
        then, by Prop.~\ref{prop:soundness},
        $T_{P,B}^{\pi}((T_{P,B}^{\pi})^n(\emptyset))
        = (T_{P,B}^{\pi})^{n+1}(\emptyset)$
        is also correct \wrt{} $P$.
    \end{itemize}
    Consequently, for all $n \in \nat$,
    $(T_{P,B}^{\pi})^n(\emptyset)$
    is correct \wrt{} $P$,
    which implies that $\patunf(P,B)$
    is correct \wrt{} $P$.
    %\qed
\end{proof}

\subsection{Proof of Theorem~\ref{theo:detection-nonterm}
(Section~\ref{sect:patterns})}
%%%%%%%%%%%%%%%%%%%%%%%%%%%%%%%%%%%%%%%%%%%
First, we establish the following result.
\begin{lemma}\label{lem:detection-nonterm-aux}
    Let $p = \patterm{u}{\sigma}{\mu}$ and
    $q = \patterm{u\sigma^a}{\sigma^b\sigma'}{\mu\mu'}$
    be two pattern terms such that $\sigma'$ commutes
    with $\sigma$ and $\mu$.
    Then, for all $n \in \nat$ there exists 
    $\eta \in \subsset$ such that
    $q(n) = p(a + n\times b)\eta$.
\end{lemma}
\begin{proof}
    For all $n \in \nat$ we have:
    \begin{align*}
        q(n) &= u\sigma^a(\sigma^b\sigma')^n\mu\mu' 
        & &\text{(by Def.~\ref{def:pattern-term})}\\
        &= (u\sigma^{a + n\times b}\mu)(\sigma')^n\mu'
        & &\text{(because $\sigma'$ commutes with $\sigma$ and $\mu$)}\\
        &= p(a + n\times b)(\sigma')^n\mu'
        & &\text{(by Def.~\ref{def:pattern-term})}
    \end{align*}
    %\qed
\end{proof}

Using Lem.~\ref{lem:detection-nonterm-aux},
we establish the next proposition.

\newpage
\begin{proposition}\label{prop:detection-nonterm-aux}
    Let $P$ be a program and $B \subseteq \patruleset$
    be correct \wrt{} $P$. Suppose that $\patunf(P,B)$
    contains a pattern rule of the form
    $(\patterm{u}{\sigma}{\mu},
    \patterm{u\sigma^a}{\sigma^b\sigma'}{\mu\mu'})$
    where $\sigma'$ commutes with $\sigma$
    and $\mu$.
    Then, for all $n \in \nat$ and all
    $\theta \in \subsset$, there exists
    $\eta \in \subsset$ such that
    $p(a + n \times b)\eta \in \calls_P(p(n)\theta)$.
\end{proposition}
\begin{proof}
    Let $n \in \nat$. 
    By Thm.~\ref{theo:soundness}, $\patunf(P,B)$ is
    correct \wrt{} $P$. Consequently, as
    $r \in \patunf(P,B)$, we have that $r$ is correct
    \wrt{} $P$, \ie $(p(n),q(n)) \in \binunf(P)$
    (Def.~\ref{def:correct}).
    
    Let $\theta \in \subsset$.
    By Prop.~\ref{prop:binunf_calls},
    there exists $\theta' \in \subsset$ such that
    $q(n)\theta' \in \calls_P(p(n)\theta)$.
    Moreover, by Lem.~\ref{lem:detection-nonterm-aux},
    there exists $\eta' \in \subsset$
    such that $q(n) = p(a + n \times b)\eta'$.
    So, we have
    $p(a + n \times b)\eta'\theta' = q(n)\theta'
    \in \calls_P(p(n)\theta)$.
    %\qed
\end{proof}

Finally, using Prop.~\ref{prop:detection-nonterm-aux},
we can now prove
Thm.~\ref{theo:detection-nonterm}.
\begin{proof}[Proof of Theorem~\ref{theo:detection-nonterm}]
    Let $p = \patterm{u}{\sigma}{\mu}$ and
    $q = \patterm{u\sigma^a}{\sigma^b\sigma'}{\mu\mu'}$.
    
    Let $n \in \nat$ and $\theta \in \subsset$.
    By Prop.~\ref{prop:detection-nonterm-aux},
    there exists $\eta_1 \in \subsset$
    such that $p(a + n \times b)\eta_1 \in \calls_P(p(n)\theta)$,
    \ie $\sequence{p(n)\theta} \mathop{\rra_P^+}
    \sequence{p(a + n \times b)\eta_1, \dots}$.
    
    Again by Prop.~\ref{prop:detection-nonterm-aux},
    there exists $\eta_2 \in \subsset$
    such that $p(a + (a + n\times b)\times b)\eta_2 \in
    \calls_P(p(a + n \times b)\eta_1)$,
    \ie $\sequence{p(a + n \times b)\eta_1} \mathop{\rra_P^+}
    \sequence{p(a + (a + n\times b)\times b)\eta_2, \dots}$.
    
    By iterating this process, we obtain an infinite
    $\rra_P$-chain: 
    \[\sequence{p(n)\theta} \mathop{\rra_P^+}
    \sequence{p(a + n\times b)\eta_1, \dots}
    \mathop{\rra_P^+}
    \sequence{p(a + (a + n\times b)\times b)\eta_2, \dots}
    \rra_P^+ \cdots\]
    %\qed
\end{proof}

\subsection{Proof of Lemma~\ref{lem:simple-trick}
(Section~\ref{sect:simple-patterns})}
%%%%%%%%%%%%%%%%%%%%%%%%%%%%%%%%%%%%%%%%%%%
Let $c \in \scontext$,
$a,b,a',b' \in \nat$ and
$u \in \stermset$.
Let $n \in \nat$. We have:
\begin{align*}
    \big(c^{a,b}(c^{a',b'}(u))\big)(n) &=
    c^{a \times n + b}(c^{a' \times n + b'}(u(n))) \\
    &= c^{(a+a') \times n + (b+b')}(u(n)) \\
    &= \big(c^{a+a',b+b'}(u)\big)(n)
\end{align*}
and
\begin{align*}
    \big(c(u)\big)(n) &= c^{0 \times n + 1}(u(n))
    = \big(c^{0,1}(u)\big)(n)
\end{align*}

\subsection{Proof of Lemma~\ref{lem:simple-pattern-term}
(Section~\ref{sect:simple-patterns})}
%%%%%%%%%%%%%%%%%%%%%%%%%%%%%%%%%%%%%%%%%%%
First, we establish the following result.
\begin{lemma}\label{lem:simple-pattern-term-aux}
    For all $s \in \stermset$,
    $\theta \in \ssubsset$ and
    $n \in \nat$ we have
    $(s\theta)(n) = s(n)\theta(n)$.
\end{lemma}
\begin{proof}
    Let $s \in \stermset$,
    $\theta \in \ssubsset$ and
    $n \in \nat$.
    By definition, $(s\theta)(n)$ is obtained from
    $s\theta$ by replacing every
    $c^{a,b} \in \Upsilon$ by
    $c^{a \times n + b}$.
    So, $(s\theta)(n) = s(n)\eta$ where, for all
    $x \in X$, $\eta(x) = (\theta(x))(n)$, \ie
    $\eta(x) = (\theta(n))(x)$ by definition of 
    $\theta(n)$. Hence, $s(n)\eta = s(n)\theta(n)$.
    Consequently, we have proved that
    $(s\theta)(n) = s(n)\theta(n)$.
    %\qed
\end{proof}

Using Lem.~\ref{lem:simple-pattern-term-aux}
we can now prove Lem.~\ref{lem:simple-pattern-term}.
\begin{proof}[Proof of Lemma~\ref{lem:simple-pattern-term}]
    Let $p = \patterm{s}{\sigma}{\mu}$ be a
    simple pattern term, $u \in \upsilon(p)$
    and $n \in \nat$.
    By Def.~\ref{def:simple-pattern-term},
    $u \in [s\theta_p]$ where, for all $x \in \vars(s)$,
    $\theta_p(x) = \mu(x)$ if $\sigma(x) = x$ and
    $\theta_p(x) = c^{a,b}(t)$ if
    $\sigma(x) = c^a(x) \neq x$ and
    $\mu(x) = c^b(t)$.

    Let $x \in \vars(s)$.
    By Def.~\ref{def:simple-pattern-term},
    $\sigma(x) = c^a(x)$ and $\mu(x) = c^b(t)$
    for some $c \in \scontext$,
    some $a,b \in \nat$ and some $t \in \termset$.
    As $\vars(c) = \emptyset$, we have
    $(\sigma^n\mu)(x) = c^{a \times n}(\mu(x))
    = c^{a \times n}(c^b(t)) 
    = c^{a \times n + b}(t)$.
    \begin{itemize}
        \item Suppose that $\sigma(x) = x$.
        Then, $\theta_p(x) = \mu(x)$.
        Hence, $(\theta_p(n))(x) = (\theta_p(x))(n) =
        (\mu(x))(n)$ with $(\mu(x))(n) = \mu(x)$
        because $\mu(x)$ contains no symbol of
        $\Upsilon$. So, $(\theta_p(n))(x) = \mu(x)$.
        Moreover, $\sigma(x) = x$ implies that
        $c = \square_1$ or $a = 0$, hence
        $\mu(x) = 
        c^{a \times n}(\mu(x)) = (\sigma^n\mu)(x)$.
        So, we have 
        $(\theta_p(n))(x) = (\sigma^n\mu)(x)$.
        \item Suppose that $\sigma(x) \neq x$.
        Then, $\theta_p(x) = c^{a,b}(t)$.
        Hence, $(\theta_p(n))(x) = (\theta_p(x))(n)
        = (c^{a,b}(t))(n)$
        with $(c^{a,b}(t))(n)
        = c^{a \times n + b}(t)$
        because $t$ contains no symbol of
        $\Upsilon$. So, $(\theta_p(n))(x) =
        c^{a \times n + b}(t) =
        (\sigma^n\mu)(x)$.
    \end{itemize}
    Therefore, we always have 
    $(\theta_p(n))(x) = (\sigma^n\mu)(x)$.
    Then, as $x$ denotes an arbitrary element of
    $\vars(s)$, we have $s(\theta_p(n)) = s\sigma^n\mu$,
    where $s\sigma^n\mu = p(n)$ (Def.~\ref{def:pattern-term}).

    Moreover, as $s$ contains no symbol from $\Upsilon$
    we have $s = s(n)$, so $s(\theta_p(n)) = s(n)(\theta_p(n))$.
    But, by Lem.~\ref{lem:simple-pattern-term-aux},
    $s(n)\theta_p(n) = (s\theta_p)(n)$ with
    $(s\theta_p)(n) = u(n)$ because
    $u \in [s\theta_p]$.
    Hence, we have $s(\theta_p(n)) = u(n)$.
    
    Finally, we have $p(n) = s(\theta_p(n)) = u(n)$.
    %\qed
\end{proof}

\subsection{Proof of Lemma~\ref{lem:simple-pattern-subs}
(Section~\ref{sect:simple-patterns})}
%%%%%%%%%%%%%%%%%%%%%%%%%%%%%%%%%%%%%%%%%%%
Let $\theta$ be a simple substitution,
$n \in \nat$ and $x \in X$.
By Def.~\ref{def:simple-pattern-subs}, we have
$\theta(x) \in [c^{a,b}(t)]$
for some $c \in \scontext$, some $a,b \in \nat$
and some $t \in \termset$. Then:
\begin{itemize}
    \item $(\theta(n))(x) = (\theta(x))(n)$;
    as $\theta(x) \in [c^{a,b}(t)]$, we have 
    $(\theta(x))(n) = (c^{a,b}(t))(n)$ with
    $(c^{a,b}(t))(n) = c^{a \times n + b}(t)$
    because $t$ contains no symbol of $\Upsilon$;
    therefore,
    $(\theta(n))(x) = c^{a \times n + b}(t)$;
    \item $\upsilon^{-1}(\theta) = \patsub{\sigma}{\mu}$
    with $\sigma(x) = c^a(x)$ and $\mu(x) = c^b(t)$;
    by Def.~\ref{def:pattern-subs},
    $(\upsilon^{-1}(\theta))(n) = \sigma^n\mu$,
    so $\big((\upsilon^{-1}(\theta))(n)\big)(x) =
    (\sigma^n\mu)(x)$ with
    $(\sigma^n\mu)(x) = c^{a \times n}(\mu(x))$
    because $\vars(c) = \emptyset$, \ie
    $(\sigma^n\mu)(x) =
    c^{a \times n}(c^b(t)) =
    c^{a \times n + b}(t)$; therefore, we have
    $\big((\upsilon^{-1}(\theta))(n)\big)(x) =
    c^{a \times n + b}(t)$.
\end{itemize}
Consequently, $(\theta(n))(x) = 
\big((\upsilon^{-1}(\theta))(n)\big)(x)$.
As $x$ denotes an arbitrary element of $X$,
we have $\theta(n) = (\upsilon^{-1}(\theta))(n)$.

\newpage
\subsection{Proof of Theorem~\ref{theo:algo-unif}
(Section~\ref{sect:pattern-unif})}
%%%%%%%%%%%%%%%%%%%%%%%%%%%%%%%%%%%%%%%%%%%
First, we establish the following result.
\begin{lemma}\label{ref:mgu-change-functions}
    For all $S,S' \in \seqset{\stermset}$,
    all $\theta \in \mgu(S,S')$
    and all $n \in \nat$, we have
    $\theta(n) \in \mgu(S(n),S'(n))$.
\end{lemma}
\begin{proof} 
    Let $S = \sequence{u_1,\dots,u_m}$ and
    $S' = \sequence{v_1,\dots,v_{m'}}$ be elements
    of $\seqset{\stermset}$.
    Suppose that $\mgu(S,S') \neq \emptyset$.
    Then, $m = m'$.
    Let $\theta \in \mgu(S,S')$ and $n \in \nat$.
    
    Let $i \in \{1,\dots,m\}$. The term $u_i(n)$
    is obtained from $u_i$ by replacing every
    $c^{a,b} \in \Upsilon$ by $c^{a \times n + b}$.
    So, considering that $c^{a,b}$ and
    $c^{a \times n + b}$ are identical symbols,
    $u_i$ and $u_i(n)$ are identical.
    Similarly, $v_i$ and $v_i(n)$ are identical.

    Hence, $S$ and $S(n)$ are identical and
    $S'$ and $S'(n)$ are identical.
    Similarly, $\theta$ and $\theta(n)$ are identical.
    So, as $\theta \in \mgu(S,S')$, we have
    $\theta(n) \in \mgu(S(n),S'(n))$.
    %\qed
\end{proof}

Using Lem.~\ref{ref:mgu-change-functions},
we can now prove Thm.~\ref{theo:algo-unif}.
For all sequences
$S = \sequence{u_1,\dots,u_m}$
of elements of $\stermset$ and
all $n \in \nat$, we let
$S(n) = \sequence{u_1(n),\dots,u_m(n)}$.
\begin{proof}[Proof of Theorem~\ref{theo:algo-unif}]
    Suppose that the algorithm successfully
    terminates from some input sequences of
    simple pattern terms $S$ and $S'$.
    Then, for some $S_1 \in \upsilon(S)$ and 
    some $S'_1 \in \upsilon(S')$,
    $\mgu(S_1,S'_1)$ contains a simple
    substitution $\theta$.
    
    Let $n \in \nat$.
    By Lem.~\ref{lem:simple-pattern-term},
    $S(n) = S_1(n)$ and $S'(n) = S'_1(n)$.
    Moreover, by Lem.~\ref{ref:mgu-change-functions},
    $\theta(n) \in
    \mgu(S_1(n),S'_1(n))$.
    Hence $\theta(n) \in \mgu(S(n),S'(n))$.
    But, by Lem.~\ref{lem:simple-pattern-subs},
    we have
    $\theta(n) = (\upsilon^{-1}(\theta))(n)$.
    So,
    $(\upsilon^{-1}(\theta))(n) \in \mgu(S(n),S'(n))$.
    
    So, as $n$ denotes an arbitray element of $\nat$,
    by Def.~\ref{def:pattern-unif} we have
    $\upsilon^{-1}(\theta) \in \mgu(S,S')$,
    where $\upsilon^{-1}(\theta)$ is the 
    pattern substitution returned by the 
    algorithm.
    %\qed
\end{proof}

\subsection{Proof of Theorem~\ref{theo:detection-nonterm-simple}
(Section~\ref{sect:nonterm-criteria})}
%%%%%%%%%%%%%%%%%%%%%%%%%%%%%%%%%%%%%%%%%%%
We use the notations of Def.~\ref{def:special-pattern-rule}.
First, we establish the following lemmas.
\begin{lemma}\label{lem:detection-nonterm-aux-0}
    Let $r$ be a special pattern rule.
    For all $n \in \nat$ such that $\alpha(r) \leq n$
    we have $0 \leq (a' - a) \times n + (d' - d) - a \times k$.
\end{lemma}
\begin{proof}
    By condition~\eqref{theo:detection-nonterm-simple-3}
    of Def.~\ref{def:special-pattern-rule}, we have $a \leq a'$.
    We distinguish two cases.
    \begin{itemize}
        \item Suppose that $a < a'$. Then,
        $\alpha(r) = \frac{a \times k - (d' - d)}{a' - a}$.
        As $\alpha(r) \leq n$, we have 
        \begin{align*}
            & \frac{a \times k - (d' - d)}{a' - a} \leq n \\
            {\ie} \quad &
            a \times k - (d' - d) \leq (a' - a) \times n \\
            {\ie} \quad &
            0 \leq (a' - a) \times n + (d' - d) - a \times k
        \end{align*}
        \item Suppose that $a = a'$. Then, 
        by condition~\eqref{theo:detection-nonterm-simple-5}
        of Def.~\ref{def:special-pattern-rule}, we have
        \begin{align*}
            & 0 \leq (d' - d) - a \times k \\
            {\ie} \quad &
            0 \leq (a' - a) \times n + (d' - d) - a \times k
        \end{align*}
    \end{itemize}
    Therefore, we always have
    $0 \leq (a' - a) \times n + (d' - d) - a \times k$.
    %\qed
\end{proof}

\begin{lemma}\label{lem:detection-nonterm-special-aux}
    Let $r = (p,q)$ be a special pattern rule.
    For all $n \in \nat$ such that $n \geq \alpha(r)$
    there exists $\eta \in \subsset$ such that $q(n) = p(n + k)\eta$.
\end{lemma}
\begin{proof}
    As $r$ is special, it is simple and 
    there exists 
    \begin{align*}
        u &=
        c\big(c_1^{a_1,b_1}(t_1),\dots,c_m^{a_m,b_m}(t_m)\big) \in \upsilon(p) \\
        v &=
        c\big(c_1^{a'_1,b'_1}(t_1\rho),\dots,c_m^{a'_m,b'_m}(t_m\rho)\big) \in \upsilon(q)
    \end{align*}
    that satisfy the conditions
    of Def.~\ref{def:special-pattern-rule}.

    Let $n \in \nat$ be such that $n \geq \alpha(r)$.
    Let $l = (a' - a) \times n + (d' - d) - a \times k$.
    By Lemma~\ref{lem:detection-nonterm-aux-0},
    we have $0 \leq l$, \ie $l \in \nat$.
    Let \[\eta = \big\{t_i \mapsto c_i^l(t_i\rho)
    \;\big\vert\; t_i \in X,\ c_i^l(t_i\rho) \neq t_i \big\}\]
    Then, $\eta$ is a substitution
    because if $t_i \in X$ and $t_i = t_j$ then $c_i = c_j$
    (condition~\eqref{theo:detection-nonterm-simple-2}
    of Def.~\ref{def:special-pattern-rule}),
    \ie $c_i^l(t_i\rho) = c_j^l(t_j\rho)$.
    
    By Lem.~\ref{lem:simple-pattern-term},
    $p(n + k) = u(n + k)$ and
    $q(n) = v(n)$. So, we have
    \begin{align*}
        p(n + k)\eta = u(n + k)\eta &=
        c\big(c_1^{a_1 \times (n + k) + b_1}(t_1\eta),\dots,
        c_m^{a_m \times (n + k) + b_m}(t_m\eta)\big) \\
        q(n) = v(n) &=
        c\big(c_1^{a'_1 \times n + b'_1}(t_1\rho),\dots,
        c_m^{a'_m \times n + b'_m}(t_m\rho)\big)
    \end{align*}
    Let $i \in \{1,\dots,m\}$.
    \begin{itemize}
        \item Suppose that $\vars(t_i) = \emptyset$. Then,
        $c_i^{a_i \times (n + k) + b_i}(t_i\eta)
        = c_i^{a_i \times (n + k) + b_i}(t_i)$
        with
        \begin{align*}
            a_i \times (n + k) + b_i &= e \times (n + k) + b \\
            &= e \times n + e \times k + b \\
            &= e \times n + (b' - b) + b \\
            &= e \times n + b' \\
            &= a'_i \times n + b'_i
        \end{align*}
        Hence, 
        $c_i^{a_i \times (n + k) + b_i}(t_i\eta)
        = c_i^{a'_i \times n + b'_i}(t_i)
        = c_i^{a'_i \times n + b'_i}(t_i\rho)$.
        \item Suppose that $t_i \in X$. Then,
        \[c_i^{a_i \times (n + k) + b_i}(t_i\eta)
        = c_i^{a \times (n + k) + d}(c_i^l(t_i\rho))
        = c_i^{a \times (n + k) + d + l}(t_i\rho)\]
        with
        \begin{align*}
            a \times (n + k) + d + l 
            &= a \times n + a \times k + d\\
            & \qquad\qquad + (a' - a) \times n + (d' - d) - a \times k \\
            &= a' \times n + d' \\
            &= a'_i \times n + b'_i
        \end{align*}
        Hence, $c_i^{a_i \times (n + k) + b_i}(t_i\eta) 
        = c_i^{a'_i \times n + b'_i}(t_i\rho)$.
    \end{itemize}
    Consequently, we have $p(n + k)\eta = q(n)$.
    %\qed
\end{proof}

Using Lem.~\ref{lem:detection-nonterm-special-aux},
we establish the next proposition.
\begin{proposition}\label{prop:detection-nonterm-special-aux}
    Let $P$ be a program and $B \subseteq \patruleset$
    be correct \wrt{} $P$. Suppose that $\patunf(P,B)$
    contains a special pattern rule $r = (p,q)$.
    Then, for all $n \in \nat$ such that $n \geq \alpha(r)$
    and all $\theta \in \subsset$, there exists
    $\eta \in \subsset$ such that
    $p(n + k)\eta \in \calls_P(p(n)\theta)$.
\end{proposition}
\begin{proof}
    Let $n \in \nat$ be such that $n \geq \alpha(r)$. 
    By Thm.~\ref{theo:soundness}, $\patunf(P,B)$ is
    correct \wrt{} $P$. As $r \in \patunf(P,B)$,
    this implies that $r$ is correct \wrt{} $P$,
    \ie $(p(n),q(n)) \in \binunf(P)$
    (Def.~\ref{def:correct}).
    
    Let $\theta \in \subsset$.
    By Prop.~\ref{prop:binunf_calls},
    there exists $\theta' \in \subsset$ such that
    $q(n)\theta' \in \calls_P(p(n)\theta)$.
    Moreover, by Lem.~\ref{lem:detection-nonterm-special-aux},
    there exists $\eta' \in \subsset$
    such that $q(n) = p(n + k)\eta'$.
    So, we have
    $p(n + k)\eta'\theta' = q(n)\theta'
    \in \calls_P(p(n)\theta)$.
    %\qed
\end{proof}

Finally, using Prop.~\ref{prop:detection-nonterm-special-aux},
we can now prove
Thm.~\ref{theo:detection-nonterm-simple}.
\begin{proof}[Proof of Theorem~\ref{theo:detection-nonterm-simple}]
    Let $n \in \nat$ be such that $n \geq \alpha(r)$. 
    Let $\theta \in \subsset$.
    By Prop.~\ref{prop:detection-nonterm-special-aux},
    there exists $\eta_1 \in \subsset$
    such that $p(n + k)\eta_1 \in \calls_P(p(n)\theta)$,
    \ie $\sequence{p(n)\theta} \mathop{\rra_P^+}
    \sequence{p(n + k)\eta_1, \dots}$.
    
    Note that $n \geq \alpha(r)$ implies that
    $n + k \geq \alpha(r)$. Hence, again by
    Prop.~\ref{prop:detection-nonterm-special-aux},
    there exists $\eta_2 \in \subsset$
    such that $p(n + k + k)\eta_2 \in \calls_P(p(n + k)\eta_1)$,
    \ie $\sequence{p(n + k)\eta_1} \mathop{\rra_P^+}
    \sequence{p(n + k + k)\eta_2, \dots}$.
    
    By iterating this process, we obtain an infinite
    $\rra_P$-chain:
    \[\sequence{p(n)\theta} \mathop{\rra_P^+}
    \sequence{p(n + k)\eta_1, \dots}
    \mathop{\rra_P^+}
    \sequence{p(n + k + k)\eta_2, \dots}
    \rra_P^+ \cdots\]
    %\qed
\end{proof}

\end{document}